\documentclass{IEEEtran}
\usepackage{cite}
\usepackage{amsmath,amssymb,amsfonts,amsthm}
\usepackage{algorithmic}
\usepackage{graphicx}
\usepackage{subfig}
\usepackage{textcomp}

\newtheorem{theorem}{Theorem}
\newtheorem{definition}{Definition}

\newtheorem{proposition}{Proposition}
\newtheorem{lemma}{Lemma}
\newtheorem{corollary}{Corollary}
\newtheorem{remark}{Remark}

\newtheorem{assumption}{Assumption}
\usepackage{svg}
\usepackage[export]{adjustbox}
\def\BibTeX{{\rm B\kern-.05em{\sc i\kern-.025em b}\kern-.08em
    T\kern-.1667em\lower.7ex\hbox{E}\kern-.125emX}}
\begin{document}
\title{On a Discrete-Time Networked SIV Epidemic Model with Polar Opinion Dynamics}
\author{Qiulin Xu and Hideaki Ishii

\thanks{Q. Xu and H. Ishii are with the Department of Computer Science, Tokyo Institute of Technology, Japan. Emails: xu.q.ab@m.titech.ac.jp, ishii@c.titech.ac.jp. This work was supported in the part by JSPS under Grant-in-Aid for Scientific Research Grant No. 22H01508.}}

\maketitle

\begin{abstract}
This paper studies novel epidemic spreading problems influenced by opinion evolution in social networks, where the opinions reflect the public health concerns. A coupled bilayer network is proposed, where the epidemics spread over several communities through a physical network layer while the opinions evolve over the same communities through a social network layer. The epidemic spreading process is described by a susceptible-infected-vigilant (SIV) model, which introduces opinion-dependent epidemic vigilance state compared with the classical epidemic models. The opinion process is modeled by a polar opinion dynamics model, which includes infection prevalence and human stubbornness into the opinion evolution. By introducing an opinion-dependent reproduction number, we analyze the stability of disease-free and endemic equilibria and derive sufficient conditions for their global asymptotic stability. We also discuss the mutual effects between epidemic eradication and opinion consensus, and the possibility of suppressing epidemic by intervening in the opinions or implementing public health strategies. Simulations are conducted to verify the theoretical results and demonstrate the feasibility of epidemic suppression.
\end{abstract}


\section{Introduction}
Mathematical modelling of infectious diseases has a long history, dating back to Daniel Bernoulli’s work on smallpox in 1760 \cite{bernoulli1760essai}. The main goals of such modelling are to understand the disease spreading mechanisms and to predict the epidemic outcome \cite{nowzari2016analysis,mei2017dynamics,pare2020modeling}. In the past century, epidemic models have been increasingly used by governments as guidance for public health policies \cite{zino2021analysis}. However, new challenges arise with virus variabilities and human complexity. The COVID-19 pandemic caused serious damages worldwide and has highlighted the significance of further research on epidemic modelling \cite{thompson2020epidemiological}.

Compartmental models, which divide populations into distinct states, are widely used in epidemiology to study the spreading dynamicsThe fundamental states are susceptible (S) and infected (I), which are present in essentially every epidemic model \cite{nowzari2016analysis}. The classical SIS model, comprising only these two states, assumes that recovered individuals do not acquire immunity to reinfection and revert to the susceptible state \cite{kermack1927contribution}. The SIR model introduces a recovered (R) state with permanent immunity, making it suitable for diseases like chickenpox \cite{brauer2008mathematical}. Networked SIS and SIR models have been extensively studied. The dynamics and convergence properties (i.e., the conditions for disease extinction) of networked SIS models with homogeneous \cite{ahn2013global,wang2021suppressing} and heterogeneous \cite{fall2007epidemiological,pare2020analysis} recovery and infection rates have been explored for both continuous- and discrete-time settings. For SIR models, where epidemics always die out eventually, researchers have focused on transient behaviors like peak infection time and level \cite{hota2021closed}, and strategies to reduce these indicators, i.e., flattening the infection curve \cite{di2020covid,wang2022resilient}. However, the recent COVID-19 pandemic showed that the SIS and SIR models are insufficient to capture certain disease characteristics and human behaviors that affect the epidemic spreading.

Some extended models have been proposed to capture more realistic epidemic spreading dynamics, such as the susceptible-alert-infected-susceptible (SAIS) model \cite{sahneh2011epidemic}, the susceptible-protected-infected-susceptible (SPIS) model \cite{theodorakopoulos2012selfish}, and the generalized susceptible-infected-vigilant (SIV) model \cite{nowzari2014stability,nowzari2015general,bhowmick2020influence}. The A, P, and V states are essentially similar, which provide individuals with temporary immunity against infection. In practice, such temporary immunity has two main sources: passive and active. The passive source is the transition from state I, which means the recovery from infection. The active source is the transition from state S, which reflects the self-protection awareness before infection, such as wearing masks, physical distancing, and getting vaccinated. The vigilant/protective state captures the impact of social awareness \cite{da2019epidemic} and virus variants \cite{holmes2021origins} on epidemic dynamics, which classical SIS and SIR models cannot study. Motivated by the practical significance, this paper focuses on the discrete-time networked SIV model.

In epidemiology, there is a key parameter called the reproduction number, which is the expected number of cases directly generated by one case in a population \cite{ma2020estimating}. The reproduction number can be affected by several factors such as environmental conditions and the behavioral patterns of the population \cite{cinelli2020covid}. In sociology, it is widely studied how opinions affect behavioral choices. In this context, the health belief model is highly relevant. It theorizes that health-related behaviors reflect both the fear of health threats and the expected fear‐reduction potential of taking actions\cite{green2020health}. This implies that the opinion evolution in a social network may affect the reproduction number, and further, the epidemic spreading, especially in this era of the Internet. A severe outbreak may induce panic opinions within some communities, and they will respond actively to control the epidemic by means of public health improvement. Conversely, disregard and denial of epidemic may lead to low vigilance and protection, and may result in persistent epidemics. As discussed above, the active source of vigilant state represents the social awareness in the SIV model, which is consistent with the health belief model. Therefore, we naturally become interested in studying the coupling between a networked SIV epidemic model and a networked opinion dynamics model. 

Studies on coupling between epidemic models and opinion dynamics have emerged as a novel research direction. Researchers in systems control have studied the opinion evolution process in social networks, and proposed several opinion dynamics models to describe the public opinion exchange \cite{proskurnikov2017tutorial}. The DeGroot model is the most fundamental one, where each individual updates his/her opinion by taking a weighted average of all neighbor's opinions \cite{degroot1974reaching}. The works \cite{lin2021discrete} and \cite{she2021peak} have respectively studied how a networked SIS and SIR models are coupled with the DeGroot model. However, the DeGroot model is over-simplified as it always leads to a consensus on a strongly connected graph, while persistent disagreements often happen in the real world. The model by Altafini in \cite{altafini2012consensus} captures the cooperative and antagonistic interactions in opinion exchange dynamics that cause disagreement, and this model has been coupled with a networked SIS model in \cite{she2022networked}. However, this model still ignores the personal preferences of individuals, which usually manifest as stubbornness or prejudice in practice \cite{frasca2013gossips}. The Friedkin–Johnsen model improved the DeGroot model by introducing an initial constant for each individual, which represents his/her stubbornness against external opinions \cite{friedkin1990social}. Polar opinion dynamics proposed in \cite{amelkin2017polar} further extended this stubbornness as a function of individual's current opinion. For instance, individuals with extreme opinions may be much harder to convince than neutral individuals. Our work in this paper uses the polar opinion dynamics to model the people's beliefs and awareness towards the severity of the epidemics.

The main contributions of this work are as follows: First, we propose a networked SIV epidemic model coupled with polar opinion dynamics. Unlike traditional epidemic models \cite{she2022networked,nowzari2014stability}, we link people’s health opinions to their health-promoting behaviors, and consider how opinions depend on peer influence, individual stubbornness, and infection levels. The epidemic spreading and opinion dynamics evolve on a bilayer multi-agent network topology, which characterizes how epidemic spreading and opinion evolution interact over a large population. Second, we define an SIV-opinion reproduction number ($R_o^V$) to measure the severity of the epidemic. Our key theoretical result is to show that if $R_o^V \leq 1$, the epidemic network goes asymptotically to the healthy disease-free equilibrium, and the opinion network asymptotically reaches a consensus that the epidemic is not a threat. For more severe cases with $R_o^V > 1$, we show some conditions under which the epidemic network will asymptotically converge to an endemic equilibrium and the opinion network will have a dissensus on the beliefs towards the severity of the epidemic. Finally, in the context of real-world public health interventions, we propose two potential ways to reduce $R_o^V$ by influencing the opinions in order to control the epidemic. Numerical simulations on a large-scale real-world network validate our results. 

This paper is organized as follows. Section \ref{Section2} introduces the preliminaries for the coupled epidemic-opinion model. Section \ref{Section3} defines the equilibria and the reproduction number of our coupled model, and analyzes the dynamical behaviors of the epidemic and opinions. A numerical example on a network of Japan's prefectures is provided in Section \ref{simulation}, which illustrates the theoretical results and explores the feasibility of controlling the epidemic by intervening in the opinions. Section \ref{Section5} concludes the paper. A preliminary version of this paper has been submitted for conference publication \cite{xu2023discrete}. The current paper contains all proofs for theoretical results, and extensive simulations are carried out as well.

\textit{Notation:} Let $[n]$ denote $\{1,2, \ldots, n\}$ for any positive integer $n$. Denote by $\mathbb{R}^n$ and $\mathbb{R}^{n \times n}$ the $n$-dimensional Euclidean space and the set of $n \times n$ real matrices, respectively. The superscript ``${\top}$'' stands for transposition of a matrix. Denote by $A\! \succ\! 0$ and $A\! \prec\! 0$ that matrix $A$ is positive definite and negative definite, respectively. Denote by $\rho(\cdot)$, $\|\cdot\|$, and $\|\cdot\|_{\infty}$ the spectral radius, Euclidean norm, and infinity norm of a matrix, respectively. The $n \times n$ identity matrix is given by $I_{n}$, and $\boldsymbol{1}_n$ represents the all-one vector in $\mathbb{R}^n$. For any matrix $M \in \mathbb{R}^{n \times n}$, $M_{ij}$ denotes its $(i,j)$-th entry. For any vector $x \in \mathbb{R}^n$, $x_i$ denotes its $i$-th entry, and $M = \operatorname{diag}\left(x\right) \in \mathbb{R}^{n \times n}$ denotes a diagonal matrix with $M_{ii}\! =\! x_i, \ \forall i \in [n]$. For any two vectors $x, y \in \mathbb{R}^n$, we simply write $x > y$ if $x_i > y_i, \forall i \in [n]$.

\section{Modelling and Problem Formulation} \label{Section2}
In this section, we introduce the networked epidemic model coupled with a polar opinion dynamics model. We consider a situation where an epidemic is spreading over a group of communities. The spreading process is affected by the topology of the physical network and the attitudes of the communities towards the diseases. On the other hand, the opinion of each community changes over time depending on its infection status and the opinions over the social network.

\subsection{Epidemic Dynamics}
In this paper, we follow the idea of the widely studied SIS model and consider its extension to a generalized one that
we are interested in.

Specifically, different from the two-state SIS model, we consider a virus spreading model with three classes of states called the SIV model \cite{nowzari2015general}. This new epidemic spreading model allows each individual to be in one of the three classes of states: Susceptible $S$, infected class $I^p$ with $p \in \{1, \ldots, m_I\}$, and vigilant class $V^q$ with $q \in \{1, \ldots, m_V\}$. A susceptible individual is capable of being infected by its infected neighbours. An infected individual recovers with a certain curing rate and become vigilant. Finally, a vigilant individual is not susceptible or infected by the disease, which can mean that this individual adopts protective actions such as wearing masks, and maintaining social distance, and becomes immune from infection or vaccination. Thus such individuals are not infectious, nor immediately susceptible to be infected. In the infected class and vigilant class, an individual can be classified into any of the $m_I/m_V$ states; this allows us to model many variations in practice, including severity of the disease (incubation or symptomatic), different sources of the vigilance (vaccines, recovery or protective actions), and so on. For simplicity, we consider the case with $m_I = m_V = 1$ in this paper. Note that based on the results of \cite{nowzari2015general}, other cases can be analyzed using a similar methodology.

Consider a physical interaction network of $n$ connected communities represented by the directed graph $\mathcal{G}_D=\left(\mathcal{V}, \mathcal{E}_D \right)$, where the node set $\mathcal{V}= [n]$ represents $n$ disjoint communities and the edge set $\mathcal{E}_D \subseteq\! \mathcal{V} \times \mathcal{V}$ represents the disease spreading interactions over $\mathcal{V}$. A directed edge $\left(j, i\right)$ indicates that community $j$ can infect community $i$. Denote by $\mathcal{N}_i^D=\left\{j \mid\left(j, i\right) \in \mathcal{E}_D\right\}$ the set of the neighbors of community $i$\footnotemark.

\footnotetext{A directed graph is used because the rates of infection can be different for a pair of communities depending on the direction (indicated by $\beta_{ij}$). Moreover, the graph itself should always be bidirectional since physical interaction among people cannot be one way.}
\begin{figure}
	\begin{center}
		\includegraphics[width=1.5in]{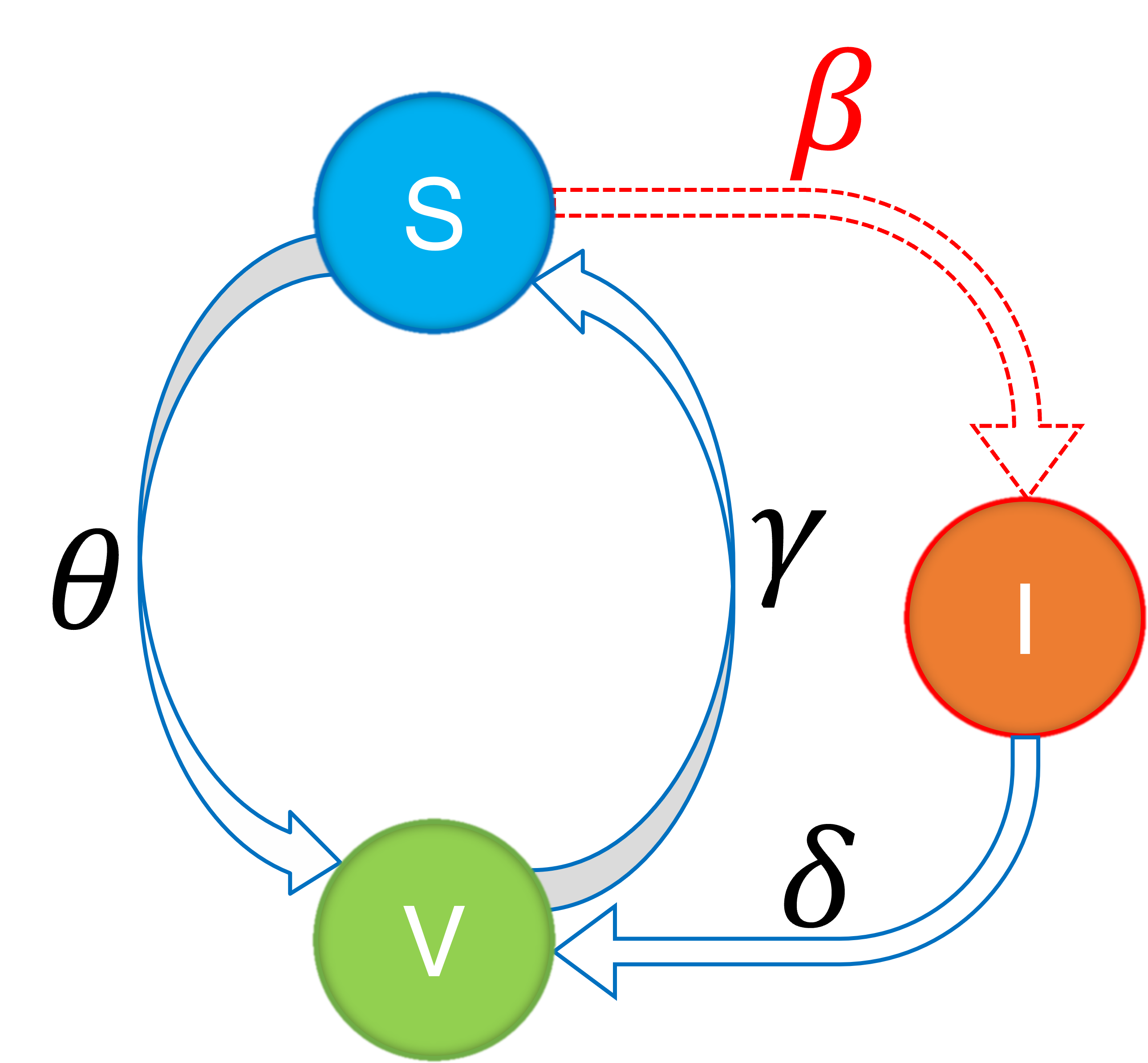}    
		\caption{SIV epidemic model with three states and various transition parameters.}  
		\label{fig0}                                 
	\end{center}                                 
\end{figure}

Fig.~\ref{fig0} shows the three-state SIV epidemic transmission model where the transitions between different compartments are shown with arrows. The proportions of the susceptible, infected, and vigilant population in community $i$ at continuous time $t \geq 0$ are denoted, respectively, by $x^S_i(t)$, $x^I_i(t)$,  and $x^V_i(t)$. Note that for all $i \in [n]$ and $t \geq 0$, it holds that $x^S_i(t), x^I_i(t), x^V_i(t) \in[0,1]$ and $x^S_i(t)+x^I_i(t)+x^V_i(t)=1$. Then, similar to the construction of the networked SIS model in \cite{nowzari2016analysis}, the SIV epidemic dynamics to capture the evolution of the $n$ communities is given by
\begin{equation}
	\begin{aligned}
		\dot{x}_i^S(t)= & \ \gamma_i x^V_i(t)-\theta_i x^S_i(t) -x^S_i(t)\sum_{j \in \mathcal{N}_i^D} \beta_{ij} x^I_j(t), \\
		\dot{x}_i^I(t)= & \ x^S_i(t)\sum_{j \in \mathcal{N}_i^D} \beta_{ij} x^I_j(t)-\delta_i x^I_i(t), \\
		\dot{x}_i^V(t)= & \ \delta_i x^I_i(t)+\theta_i x^S_i(t)-\gamma_i x^V_i(t). \label{SIV-continuous}
	\end{aligned}
\end{equation}
The transition parameters are given as follows: $\beta_{ij} \in[0,1]$ denotes the average infection rate from community $j$ to community $i$; $\delta_i \in[0,1]$ denotes the average recovery rate of the infected population in community $i$; $\gamma_i \in[0,1]$ is the average susceptibility rate of the vigilant population in community $i$ after despising protective measures or losing immunity; and $\theta_i \in[0,1]$ is the average vigilance rate of the susceptible population in community $i$ to become vigilant.

In this paper, we deal with the system in the discrete-time domain as in the discrete-time SIS model case in \cite{pare2020analysis}. Based on Euler's method, we describe (\ref{SIV-continuous}) in the approximate form with the time index $k\geq 0$ as as 
\begin{equation}
	\begin{aligned}
		x_i^S(k+1)= & \ x_i^S(k) + \gamma_i x^V_i(k)-\theta_i x^S_i(k) \\
		&\ -x^S_i(k)\sum_{j \in \mathcal{N}_i^D} \beta_{ij} x^I_j(k), \\
		x_i^I(k+1)= & \ x_i^I(k)+ x^S_i(k)\sum_{j \in \mathcal{N}_i^D} \beta_{ij} x^I_j(k)-\delta_i x^I_i(k), \\
		x_i^V(k+1)= & \ x_i^V(k) + \delta_i x^I_i(k)+\theta_i x^S_i(k)-\gamma_i x^V_i(k). \label{SIV-discrete}
	\end{aligned}
\end{equation}
Without loss of generality, we take sampling period $\Delta T = 1$ in this paper. In fact, if $\Delta T$ is not $1$, the value of parameters $\beta_{ij}, \delta_{i}, \gamma_{i}, \theta_{i}$ can be replaced by $\beta_{ij} \Delta T, \delta_{i} \Delta T, \gamma_{i} \Delta T, \theta_{i} \Delta T$ and then the dynamics is still (\ref{SIV-discrete}).

\subsection{Opinion Dynamics}
For community $i$, its opinion towards the epidemic severity  at time $k$ is denoted as $o_i(k) \in[0,1], i \in [n]$. With $o_i(k)=1$, community $i$ believes that the epidemic is extremely serious, and with $o_i(k)=0$, community $i$ perceives that the epidemic does not pose any threat. 

In general, a community may not always hold its original opinion in a social network. Each community evaluates and changes its opinion by comparison with the opinions of its neighbors. Consider the opinions evolving through a social network of $n$ connected communities, which is represented by the directed graph $\mathcal{G}_O=\left(\mathcal{V}, \mathcal{E}_O \right)$. Similar to the definition of the physical network $\mathcal{G}_D$, the neighbor set of community $i$ in this social network $\mathcal{G}_O$ is defined as $\mathcal{N}_i^O$.

In this paper, we focus on the following polar opinion dynamics with stubborn positives:
\begin{equation} \label{opinion} 
	o_i(k+1)=o_i(k) + (1-o_i(k))\sum_{j \in \mathcal{N}_i^O} w_{i j} (o_j(k)-o_i(k)), 
\end{equation}
where the weight $w_{i j} \geq 0$ measures the amount of relative influence of community $j$ upon community $i$. Assume that $\sum_{j} w_{i j} = 1$ for all $i \in[n]$. Then let $W=\left[w_{i j}\right] \in \mathbb{R}^{n \times n}$ be the row-stochastic adjacency matrix of the social network, and let $L=I_n-W$ be the network's Laplacian matrix.

This model (\ref{opinion}) considers the human stubbornness in opinion dynamics. Note that the main difference between the dynamics (\ref{opinion}) and the classical DeGroot model\! \cite{degroot1974reaching} is the term $1-o_i(k)$. This term implies that extreme opinions on one end of the spectrum are more resistant to changes than the other end. This model captures the cases when the communities at one negative extreme of the opinion spectrum may be more open to the opposite opinion, while the communities with the opposite, positive, opinions have an incentive to maintain their positions \cite{amelkin2017polar}. In this paper, we assume that $o_i(k)=1$ represents stubborn positives based on the theory of mass panic \cite{mawson2017mass}.

\subsection{Coupled Epidemic-Opinion Dynamics}
After introducing the networked SIV epidemic model and the polar opinion model spreading over the same set of $n$ communities, it is natural to consider a network dynamical model that couples the two models together. 

First, based on the health belief model in health behavior research \cite{green2020health}, it is reasonable to expect that a community’s opinion or attitude toward the epidemic severity will affect its actions of adopting protective behaviors. For example, a community being very serious about the epidemic may propagate the dangers of the epidemic more widely and frequently, and tend to make stricter policies to prevent the epidemic. People in such a community may also be more likely to adhere to protective behaviors, follow scientific instructions, and get vaccinated actively. These actions may lead to higher vigilance rate $\theta$ and lower susceptibility rate $\gamma$ in the community. To describe this dependence, we take the rates $\theta_i$ and $\gamma_i$ to be functions of opinion $o_i(k)$ as $\theta_i(o_i(k))$ and $\gamma_i(o_i(k))$. 

\textit{Analysis of $\theta_i(o_i(k))$ and $\gamma_i(o_i(k))$:} We do not give any specifics about mathematical forms of $\theta_i(o_i(k))$ and $\gamma_i(o_i(k))$ here. They may vary in different social environments and need a case-by-case analysis based on sociopsychological studies \cite{bechler2021attitude,ajzen1977attitude,verplanken2022attitudes}. For example, with the same opinion level $o_i(k)$, a community with better sanitary conditions and sounder public health policies will have larger $\theta_i(o_i(k))$ and smaller $\gamma_i(o_i(k))$ than others. Nevertheless, despite the lack of mathematical forms, our results in this paper are derived independently of any particular mathematical forms, only relying on the value ranges of $\theta_i(o_i(k))$ and $\gamma_i(o_i(k))$ as well as the property that $\gamma_i(o_i(k)) + \theta_i(o_i(k))$ is strictly larger than $0$ for all $k \geq 0$. The latter condition is reasonable since $\gamma$ and $\theta$ should generally have inverse gradients with respect to $o$ in reality.

On the other hand, the proportion of infection in a community can certainly affect its opinion on how severe the epidemic is. Consider the following opinion dynamics model of community $i$, which captures the original opinion model (\ref{opinion}) and its current infection level in (\ref{SIV-discrete}):
\begin{equation} 
	\begin{aligned}
	o_i(k+1)=&\ \phi_i x_i^I(k) + (1-\phi_i)\left[o_i(k) \right.\\
	&+ (1-o_i(k))\sum_{j \in \mathcal{N}_i^O} w_{i j} (o_j(k)-o_i(k))\left.\right], \label{opinion-coupled}
	\end{aligned}
\end{equation}
where $\phi_i \in(0,1)$ is a given constant. The second term on the right-hand side of (\ref{opinion-coupled}) is from (\ref{opinion}). The neighbors of community $i$ influence its opinion following the polar model with stubborn positives. The first term captures the impact of community $i$'s infection level on its opinion. For instance, if $o_i(k+1)$ obtained by (\ref{opinion}) is small but the community is severely infected, i.e., $x_i^I(k)$ is large, then $o_i(k+1)$ will increase in (\ref{opinion-coupled}). This model is consistent with the health belief model in \cite{green2020health}.

\subsection{Problem Statement}
In epidemiology, the basic reproduction number $R_0$ is the key metric to describe cases directly generated by one case without intervention. Now that we have proposed a coupled SIV-opinion model (\ref{SIV-discrete}) and (\ref{opinion-coupled}), we proceed to introduce an opinion-dependent reproduction number to describe the epidemic spreading in this model. Specifically, we will analyze the system equilibria and their stability conditions.

\section{Analysis of SIV-Opinion Dynamical Model} \label{Section3}
This section considers well-posedness and the equilibria of the SIV-opinion dynamical model. Furthermore, we analyze stability conditions of our model and discuss the mutual influence between epidemic spreading and opinion evolution.

\subsection{Well-Posedness} \label{section3A}
For our coupled epidemic-opinion model to be well posed, its solutions must remain in the state space $[0,1]^{n}$. To this end, we pose three assumptions related to the graphs, transition parameters, and initial states.

\begin{assumption} \label{connection}
	Both the physical interaction graph $\mathcal{G}_D$ and the social graph $\mathcal{G}_O$ are strongly connected. \label{assumption3}
\end{assumption}

\begin{assumption}
	For all $i, j \in [n]$ and $k \geq 0$, it holds that $\delta_i, \beta_{ij}, \theta_i(o_i(k)), \gamma_i(o_i(k)) \in [0,1]$, \!$\sum_{j \in \mathcal{N}_i^D} \beta_{ij} + \theta_i(o_i(k)) \leq 1$, and $\theta_i(o_i(k))+ \gamma_i(o_i(k)) \geq c$ for some constant $c \in (0, 1)$. \label{assumption2}
\end{assumption}
\vspace{-0.6cm}
\begin{assumption} 
	For all $i \in[n]$, it holds that $x_i^S(0), x_i^I(0),$ $x_i^V(0), o_i(0) \in[0,1]$ and $x^S_i(0) + x^I_i(0) + x^V_i(0) = 1$. \label{assumption1}
\end{assumption}

Under the assumptions, we have the following proposition.
\begin{proposition} \label{well posedness}
	For the model in (\ref{SIV-discrete}) and (\ref{opinion-coupled}), the states satisfy $x_i^S(k), x_i^I(k),$ $x_i^V(k), o_i(k) \in[0,1]$ for all $i \in[n]$ and $k \geq 0$.
\end{proposition}
\begin{proof}
	We show the well-posedness of the coupled epidemic-opinion model by induction. Suppose that at time $k$, $x_i^S(k), x_i^I(k),$ $x_i^V(k), o_i(k) \in[0,1]$ and $x^S_i(k) + x^I_i(k) + x^V_i(k) = 1$. Rearranging (\ref{SIV-discrete}), we have
	\begin{equation}
		\begin{aligned}
			x_i^S(k+1)= & \ \bigg(1 - \sum_{j \in \mathcal{N}_i^D} \beta_{ij} x^I_j(k) - \theta_i(o_i(k))\bigg) x_i^S(k) \\
			& \ + x_i^V(k)\gamma_i(o_i(k)),\\
			x_i^I(k+1)= & \ (1-\delta_i)x_i^I(k) + x^S_i(k)\sum_{j \in \mathcal{N}_i^D} \beta_{ij} x^I_j(k), \\
			x_i^V(k+1)= & \ (1-\gamma_i(o_i(k)))x_i^V(k) \\
			& \ + \delta_i x^I_i(k)+\theta_i(o_i(k)) x^S_i(k). 
		\end{aligned} \label{posedness}
	\end{equation}
By Assumptions \ref{assumption2} and \ref{assumption1}, it holds $x_i^S(k+1), x_i^I(k+1), x_i^V(k+1) \geq 0$. Moreover, by summing the left- and right-hand sides of the three equations in (\ref{posedness}) respectively, we obtain $x^S_i(k+1) + x^I_i(k+1) + x^V_i(k+1) = x^S_i(k) + x^I_i(k) + x^V_i(k) = 1$. Therefore, $x_i^S(k+1), x_i^I (k+1), x_i^V(k+1) \in [0,1]$.


Then consider the range of $o_i(k+1)$. From (\ref{opinion-coupled}), we have
$$
	\begin{aligned}
		o_i(k+1)=\ \phi_i x_i^I(k) + (1-\phi_i)\bar{o}_i(k+1),
	\end{aligned}
$$
where 
$$
\bar{o}_i(k+1) = o_i^2(k) \\
+ (1-o_i(k))\big(o_i(k) + \sum_{j \in \mathcal{N}_i^O} w_{i j} (o_j(k)-o_i(k))\big).
$$
Since $W$ is a row-stochastic matrix, we have $o_i(k) + \sum_{j \in \mathcal{N}_i^O} w_{i j} (o_j(k)-o_i(k)) \in [0,1]$. Note that $\bar{o}_i(k+1)$ is a convex combination of $o_i(k)$ and $o_i(k) + \sum_{j \in \mathcal{N}_i^O} w_{i j} (o_j(k)-o_i(k))$, and $o_i(k+1)$ is also a convex combination of $x_i^I(k)$ and $\bar{o}_i(k+1)$. Therefore, it holds $o_i(k+1) \in [0,1]$.
\end{proof}

\subsection{Equilibria of the Coupled Model}
Due to the constraint that $x^S_i(k) + x^I_i(k) + x^V_i(k) = 1$ for all $i \in[n]$ and $k \geq 0$, one of the equations in (\ref{SIV-discrete}) is redundant. By setting $x^S_i(k) = 1- x^I_i(k) - x^V_i(k)$, the coupled SIV-opinion model can be described by
\begin{equation}
	\begin{aligned}
		x_i^I(k+1)= & \ x_i^I(k) - \delta_i x^I_i(k) \\
		&\ + (1- x^I_i(k) - x^V_i(k))\sum_{j \in \mathcal{N}_i^D} \beta_{ij} x^I_j(k), \\
		x_i^V(k+1)= & \ x_i^V(k) + \delta_i x^I_i(k) - \gamma_i(o_i(k)) x^V_i(k) \\
		& \ + \theta_i(o_i(k)) (1- x^I_i(k) - x^V_i(k)), \\
		o_i(k+1)=&\ \phi_i x_i^I(k) + (1-\phi_i)\left[o_i(k) \right.\\
		&+ (1-o_i(k))\sum_{j \in \mathcal{N}_i^O} w_{i j} (o_j(k)-o_i(k))\left.\right]. \label{coupled model}
	\end{aligned}
\end{equation}
To study the system (\ref{coupled model}), let (${x_i^I}^*$, ${x_i^V}^*$, $o_i^*$) denote an equilibrium state of the three equations.

\begin{definition}
	An equilibrium state $z^* = ({x^I}^*, {x^V}^*,o^*)$ of the coupled SIV-opinion model (\ref{coupled model}) is said to be 
	\begin{enumerate}
		\item a healthy state if ${x^I}^*=0$, and an endemic state otherwise;
		\item a consensus state if $o_i^* = o_j^*, \forall i, j \in[n]$, and a dissensus state otherwise.
	\end{enumerate}
\end{definition}

\subsection{Stability Analysis of Disease-Free Equilibrium} \label{Disease-Free Equilibrium}
In practice, achieving the healthy state, which means to reach a disease-free equilibrium, should be the most worth exploring scenario. To further analyze stability conditions of disease-free equilibria, we state a few preliminaries.

In epidemiology,  the basic reproduction number, denoted
by $R_0$, is a critical parameter to measure the spreading of
an epidemic \cite{ma2020estimating}. It can be affected by other factors such as pathogen types and the behaviours of the population. Thus, we define a specific reproduction number $R_o^V$ to characterize the infectivity of the SIV-opinion model (\ref{coupled model}).
\begin{definition}[SIV-opinion reproduction number] \label{definition}
	For the coupled SIV-opinion model in (\ref{coupled model}), the reproduction number is defined as
	$$ 
		R_o^V = \rho\left(I_n - \Delta + B - \Psi B\right),
	$$
	where 
	\begin{align} 
		\Psi &= \operatorname{diag}\left( \min _{o_i \in [\underline{o}_i, \overline{o}_i]}  \psi_i(o_i)\right), \label{Psi} \\
		\psi_i(o_i) &=  \frac{\theta_i(o_i)}{\gamma_i(o_i) + \theta_i(o_i)}, i \in[n], \label{chi}
	\end{align}
	$B=\left[\beta_{i j}\right] \in \mathbb{R}^{n \times n}$, and $D=\operatorname{diag}\left(\delta_1, \ldots, \delta_n\right)$, with $\underline{o}_i$ and $\overline{o}_i$ being the lower and upper bounds of $o_i(k)$ for $k \geq 0$, respectively.
\end{definition}
Compared to classical networked SIS models \cite{nowzari2016analysis, pare2020analysis}, in which reproduction numbers depend only on infection and recovery rates, Definition \ref{definition} incorporates the influence of opinions by introducing the parameter $\Psi$ into the reproduction number. Equations (\ref{Psi}) and (\ref{chi}) indicate that, throughout opinion evolution, higher level of vigilance towards the epidemic, i.e., larger $\theta(o)$ or smaller $\gamma(o)$, leads to a smaller reproduction number $R_o^V$, and vice versa. We then proceed to analyze the system behavior in (\ref{coupled model}) through bounds on $R_o^V$. The following lemma will be employed in deriving the main result later.
\begin{lemma} \label{lemma-1}
	For the the coupled SIV-opinion model in (\ref{coupled model}), $o(k)$ asymptotically converges to $0$ if $x^I(k)$ asymptotically converges to $0$.
%
\end{lemma}
\begin{proof}
	We write the third equation of (\ref{coupled model}) in a compact form:
	\begin{equation} \label{compact o}
		o(k+1) = \Phi x^I(k) + (I_n - \Phi) \widetilde{W}(o(k)) o(k),
	\end{equation}
	where
	$$
	\widetilde{W}(o(k)) = W + O(k) L  = (I_n - O(k))W + O(k),
	$$
	$\Phi = \operatorname{diag}\left(\phi_1, \ldots, \phi_n\right)$, and $O(k)=\operatorname{diag}\left(o_1(k), \ldots, o_n(k)\right)$. Since the adjacency matrix $W$ is row stochastic and $o_i(k) \in[0,1]$, $\widetilde{W}(o(k))$ is also row stochastic. 
	
	We now consider the stability of the system
	\begin{equation} \label{o_bar}
		\tilde{o}(k+1) = (I_n - \Phi) \widetilde{W}(\bar{o}(k)) \tilde{o}(k).
	\end{equation}
	We use $\tilde{o}_{\max}(k)$ and $\phi_{\min}$ to denote $\max _{i \in [n]} \tilde{o}_i(k)$ and $\min _{i \in [n]} \phi_i$, respectively. From the row-stochasticity of $\widetilde{W}(o(k))$, it follows that
	$$
		\tilde{o}_{\max}(k+1) \leq (1-\phi_{\min}) \tilde{o}_{\max}(k).
	$$
	Since $\phi_i \in(0,1)$, $\tilde{o}_{\max}(k)$ will exponentially converge to $0$ for all initial values. Consequently, $\tilde{o}(k)$ in (\ref{o_bar}) will exponentially converge to $0$ for all initial conditions in $[0,1]^n$ from the well-posedness result of Theorem \ref{well posedness}. 
	
	Then, consider $x^I(k)$ as the external input of system (\ref{compact o}). It follows immediately from the stability of (\ref{o_bar}) that (\ref{compact o}) is input-to-state stable. Hence, with $x^I(k)$ converging to $0$ asymptotically, $o(k)$ asymptotically goes to $0$ in (\ref{compact o}).
\end{proof}

\begin{remark}
	Lemma \ref{lemma-1} reveals a social state that corresponds to the extinction of the epidemic. When the epidemic dies out for whatever reason, the society will end with a consensus-healthy equilibrium $z^* = (0, {x^V}^*,0)$. This situation implies that, if the epidemic is eradicated, all communities will reach an agreement that the epidemic poses no threat.
\end{remark}

Now we are ready to prove the stability of the disease-free equilibrium. First, we define a particular equilibrium of the vigilant state as
$$
	{\hat{V}}^* = \Theta(0)(\Gamma(0)+\Theta(0))^{-1} \boldsymbol{1}_n,
$$
where $\Theta(o) = \operatorname{diag}\left(\theta_1(o_1), \ldots, \theta_n(o_n)\right)$ and $\Gamma(o) = \operatorname{diag}\left(\gamma_1(o_1), \ldots, \gamma_n(o_n)\right)$. Then in the following theorem, a sufficient condition for the global stability of the disease-free equilibrium will be established. To simplify notations, $\gamma_i(o_i(k))$, $\theta_i(o_i(k))$, $\Gamma(o(k))$, $\Theta(o(k))$ and $\gamma_i(k)$, $\theta_i(k)$, $\Gamma(k)$, $\Theta(k)$ can substitute each other in the rest of this paper. 
\begin{theorem} \label{main result-health}
	If $R_o^V \leq 1$, the healthy-consensus state $z^* = (0, {\hat{V}}^*, 0)$ of the system in (\ref{coupled model}) is globally asymptotically stable.
\end{theorem}
\begin{proof}
Note that at the disease-free equilibrium, $x_i^V(k)$ does not necessarily go to $0$ as it depends on the parameters $\gamma_i(k)$ and $\theta_i(k)$. Thus, its corresponding equilibrium ${x_i^V}^*$ is obtained by 
$$
0 = \delta_i {x^I_i}^* - \gamma_i(o_i^*) {x^V_i}^* \ + \theta_i(o_i^*) (1- {x^I_i}^* - {x^V_i}^*),
$$
and then we have 
\begin{equation} \label{xVstar}
	{x^V_i}^* =\frac{\delta_i {x^I_i}^* + \theta_i(o_i^*) (1-{x^I_i}^*)}{\theta_i(o_i^*)+\gamma_i(o_i^*)}.
\end{equation}
Therefore, to simplify the following stability analysis and to replace $x^V_i(k)$, we introduce the new state
\begin{equation}
	\label{coordinate change}
	{e^V_i}(k) = x^V_i(k) - \frac{\delta_i x^I_i(k) + \theta_i(k) \left(1-x^I_i(k)\right)}{\theta_i(k)+\gamma_i(k)}.
\end{equation}
From (\ref{coupled model}) and (\ref{coordinate change}),
$$
\begin{aligned}
	&{e^V_i}(k+1) =  \ x^V_i(k+1) \\
	& \quad\quad\quad\quad\quad - \frac{\delta_i x^I_i(k+1) + \theta_i(k+1) \left(1-x^I_i(k+1)\right)}{\theta_i(k+1)+\gamma_i(k+1)} \\
	&=  \left(1-\gamma_i(k)-\theta_i(k)\right) x_i^V(k)  \\
	& \ + \left[\delta_i - \theta_i(k) + \frac{\left(\theta_i(k+1)-\delta_i\right)(1-\delta_i)}{\theta_i(k+1)+\gamma_i(k+1)}\right] x^I_i(k) \\
	& \ + \frac{\left(\theta_i(k+1)-\delta_i\right)\left(1-x^I_i(k)-x^V_i(k)\right)}{\theta_i(k+1)+\gamma_i(k+1)}  \sum_{j \in \mathcal{N}_i^D} \beta_{ij} x^I_j(k) \\
	& \ + \theta_i(k) - \frac{\theta_i(k+1)}{\theta_i(k+1)+\gamma_i(k+1)} \\
	&= \left(1-\gamma_i(k)-\theta_i(k)\right) e^V_i(k) \\
	& \ + \left[\frac{\delta_i - \theta_i(k)}{\theta_i(k)+\gamma_i(k)} + \frac{\left(\theta_i(k+1)-\delta_i\right)(1-\delta_i)}{\theta_i(k+1)+\gamma_i(k+1)}\right] x^I_i(k) \\
	& \ + \frac{\left(\theta_i(k+1)-\delta_i\right)\left(1-x^I_i(k)-x^V_i(k)\right)}{\theta_i(k+1)+\gamma_i(k+1)}  \sum_{j \in \mathcal{N}_i^D} \beta_{ij} x^I_j(k) \\
	& \ + \frac{\theta_i(k)}{\theta_i(k)+\gamma_i(k)} - \frac{\theta_i(k+1)}{\theta_i(k+1)+\gamma_i(k+1)}.
\end{aligned}
$$
This can be rewritten in the following compact form:
\begin{equation} \label{compact-xe}
	\begin{aligned}
		{e^V}(k+1) =& \left(I_n-\Gamma(k)-\Theta(k)\right) e^V(k) \\
		&+ \Xi(k) x^I(k) - \Upsilon(k)x^I(k) + \kappa(k),
	\end{aligned}
\end{equation}
where 
$$
\begin{aligned}
	\Xi(k) =& \ (\Delta - \Theta(k))(\Theta(k)+\Gamma(k))^{-1} \\ &\!\!+(\Theta(k+1)-\Delta)(I_n-\Delta)(\Theta(k+1)+\Gamma(k+1))^{-1} \\ &\!\!+(\Theta(k+1)-\Delta)(\Theta(k+1)+\Gamma(k+1))^{-1}B, \\
	\Upsilon(k) =& \ (\Theta(k+1)-\Delta)(\Theta(k+1)\\
	&\!\!+\Gamma(k+1))^{-1} \left(x^I(k)+x^V(k)\right)B, \\
	\kappa(k) =& \ (\Theta(k)+\Gamma(k))^{-1} \theta(k)\\
	&\!\!- (\Theta(k+1)+\Gamma(k+1))^{-1} \theta(k+1).
\end{aligned}
$$

Then, using (\ref{coordinate change}) in the first equation of (\ref{coupled model}), one obtains
$$
\begin{aligned}
	{x^I_i}(k+1) =& \ (1-\delta_i)x_i^I(k) + \bigg(1 - \frac{\theta_i(k)}{\theta_i(k)+\gamma_i(k)} \\
	&\!\! - \frac{\delta_i+\gamma_i(k)}{\theta_i(k)+\gamma_i(k)} {x^I_i}(k) -  {e^V_i}(k) \bigg)\!\sum_{j \in \mathcal{N}_i^D}\! \beta_{ij} x^I_j(k),
\end{aligned}
$$
which can be written compactly as
\begin{equation} \label{compact-xI}
	{x^I}(k+1) = \mathcal{R}(k)x^I(k) - \zeta(k) - \iota(k),
\end{equation}
where $\mathcal{R}(k)\! =\! I_n - \Delta + (I_n - \Theta(k)(\Gamma(k)+\Theta(k))^{-1}) B, \zeta(k) \!=\! \operatorname{diag} \left(x^I(k)\right)Bx^I(k)$, and $\iota(k) = \operatorname{diag} \left(x^V(k)\right)Bx^I(k)$.

From Proposition \ref{well posedness}, it follows that $0 \leq \iota(k) \leq B x^I(k)$. Consequently, the following model is considered:
\begin{equation} \label{linear system}
	{\bar{x}^I}(k+1)  = \bar{\mathcal{R}} {\bar{x}^I}(k) -  \zeta(k),
\end{equation}
where $\bar{\mathcal{R}} = I_n - \Delta +( I_n - \Psi) B$. By (\ref{Psi}), we have $0\leq {x^I}(k) \leq {\bar{x}^I}(k)$ for all $k$. Recalling Definition \ref{definition}, we have $\rho\left(\bar{\mathcal{R}}\right) = R_o^V$. Now, the stability analysis in Theorem 1 of \cite{pare2020analysis} can be applied to the system (\ref{linear system}). Then, it follows from $R_o^V \leq 1$ that $\bar{x}^I(k)$ asymptotically converges to $0$ for any initial state ${\bar{x}^I}(0) \in [0,1]^n$. This also implies that $\iota(k)$ asymptotically converges to $0$. Therefore, we have that (\ref{compact-xI}) is asymptotically stable for any initial state ${x^I}(0), {x^V}(0) \in [0,1]^n$.

It remains to show that ${e^V}(k)$ goes to $0$ asymptotically, and we consider its dynamics in (\ref{compact-xe}). From Lemma \ref{lemma-1}, $\kappa(k)$ will asymptotically converge to $0$ if ${\bar{x}^I}(k)$ asymptotically converges to $0$. Since $\inf _{\substack{k \geq 0 \\ i \in [n]}} \{\gamma_i(k) + \theta_i(k)\} \geq c$ from Assumption \ref{assumption2}, ${e^V}(k+1) = \left(I_n-\Gamma(k)-\Theta(k)\right) e^V(k)$ will exponentially converge to $0$. Note that $\Xi(k)$ and $\Upsilon(k)$ both take finite values because of the boundedness of the transition parameters and states. Therefore, (\ref{compact-xe}) is input-to-state stable. It follows that ${e^V}(k)$ will asymptotically converge to $0$ for all initial conditions with input ${x^I}(k)$ vanishing asymptotically. The statement of the theorem then follows immediately.
\end{proof}
 
Theorem \ref{main result-health} shows the role of the reproduction number $R_o^V$, or more specifically, the lower bound of $\psi_i(o_i(k))$ defined in (\ref{chi}), in epidemic eradication.  In practice, when $R_o^V$ is large and the epidemic cannot disappear spontaneously, administrations of the communities can lead the population so that the lower bound of $\psi(o(k))$ becomes larger, which will make $R_o^V$ smaller than $1$. Numerical examples will be provided in Section \ref{simulation} to illustrate the effectiveness of such control strategies. Further, it will be interesting to analyze optimal control strategies theoretically to realize effective epidemic suppression under some budget constraints in the future.

\subsection{Stability Analysis of Endemic Equilibrium}
Since we analyzed the stability of the disease-free equilibrium, the next step is to study the endemic equilibrium with ${x^I}^*_i>0$ for all $i\in[V]$, which reveals the impact of opinions under situations with more severe epidemics. 

We have seen in Lemma \ref{lemma-1} that all communities reach consensus when the epidemic disappears. Now we consider the opinion states of the endemic equilibrium. Denote by $\mathcal{Z}$ the set of all endemic equilibria for the coupled SIV-opinion system (\ref{coupled model}). We have the following proposition.

\begin{proposition} \label{proposition1}
	For the coupled SIV-opinion system in (\ref{coupled model}), an consensus-endemic state (with ${x^I}^* \neq 0 $, $o^* = a \boldsymbol{1}_n$) is an equilibrium only if $a \in (0,1]$ and 
	\begin{equation} \label{curve}
		\delta_i = \frac{1 - a - \frac{\theta_i(a)(1-a)}{\gamma_i(a)+\theta_i(a)}}{\frac{a}{\gamma_i(a)+\theta_i(a)} + \frac{1}{\sum_{j \in \mathcal{N}_i^D} \beta_{ij}}}, \ \forall i \in[n].
	\end{equation}
\end{proposition}
\begin{proof}
	Substituting $o^* = a \boldsymbol{1}_n$ into the third equation of (\ref{coupled model}), we have
	$$
	o_i^*=\phi_i {x_i^I}^* +\left(1-\phi_i\right) o_i^*, \ \forall i \in[n].
	$$
	Thus we have ${x_i^I}^* = a, \forall i \in[n]$. Since ${x^I}^* \neq 0 $, one obtains $a \in (0,1]$. Then we derive from the first two equations of (\ref{coupled model}) and (\ref{xVstar}) that
	$$
		\frac{\delta_i}{\sum_{j \in \mathcal{N}_i^D} \beta_{ij}} + \frac{\delta_i a + \theta_i(a) (1-a)}{\gamma_i(a) + \theta_i(a)} + a = 1, \ \forall i \in[n].
	$$
	We then obtain (\ref{curve}).
\end{proof}

The following corollary is a direct result from the system model of (\ref{coupled model}) and Proposition \ref{proposition1}.
\begin{corollary} \label{corollary-endemic}
	If $z^*$ is an endemic equilibrium of the coupled SIV-opinion model in (\ref{coupled model}), then ${x^I}^* \!> 0$, ${x^V}^* \!> 0$, and $o^* > 0$.
\end{corollary}

Proposition \ref{proposition1} and Corollary \ref{corollary-endemic} state that as long as the epidemic persists, no community can be completely disease-free or agree that the epidemic does not pose a threat. Furthermore, under such a situation, it is very unlikely that the communities reach a consensus on the severity of the epidemic since it will require the transition parameters to satisfy the specific relation stated in (\ref{curve}).

We would like to study the stability of the endemic equilibrium. First, let us define 
$$
\begin{aligned}
	\eta & \triangleq \max _{\substack{o_i \in [\underline{o}_i, \overline{o}_i], i \in [n]}} \{1-\gamma_i(k) - \theta_i(k)\}, \\
	\nu & \triangleq \max _{\substack{o_i \in [\underline{o}_i, \overline{o}_i], i \in [n]}} \{\theta_i(k)-\delta_i\},
\end{aligned}
$$
and for a given endemic equilibrium $z^*$, let
\begin{equation} \label{phi}
\varphi = \max _{\substack{x^S \in [0, 1]^n}} \left\|	I_n - \Delta - \mathcal{B}^* + \operatorname{diag}\left(x^S\right)B\right\|_{\infty},
\end{equation}
where $\mathcal{B}^* = \operatorname{diag}(B{x^I}^*)$. Then, the following result characterizes the condition under which the endemic equilibrium is globally asymptotically stable.

\begin{theorem} \label{theorem-endemic}
	Suppose that $R_o^V > 1$ and $z^* = ({x^I}^*, {x^V}^*,o^*)$ is an endemic equilibrium of the coupled SIV-opinion model (\ref{coupled model}). Then $z^*$ is  asymptotically stable for all disease-nonzero initial conditions, i.e., $x^I(0) \neq 0$, if 
	\begin{equation} \label{range}
		  \sum_{j \in \mathcal{N}_i^D}\! \beta_{ij} < \delta_i + \mathcal{B}^*_{ii} < 2-\!\sum_{j \in \mathcal{N}_i^D}\! \beta_{ij}+2\beta_{ii}, 
	\end{equation}
	\vspace{-0.3cm}
	\begin{equation} \label{range-1}
		-2w_{ii}-\frac{\phi_i}{1-\phi_i} < \mathcal{O}^*_{ii} < \frac{\phi_i}{1-\phi_i}, 
	\end{equation}
	for all $i \in[n]$, where 
	$$\mathcal{O}^* = \operatorname{diag}\left(L o^*\right),$$
	and there exist $\varsigma_1, \varsigma_2 >0$ such that
	\begin{equation} \label{theorem2}
		\begin{aligned}
			\frac{\varsigma_2 \nu^2}{1-\eta^2} + \nu^2 + \frac{\varsigma_1^2 \varphi^2}{(1-\varphi^2)^2} & < \varsigma_1, \\
			\frac{\varsigma_1 \rho^2(\mathcal{B}^*)}{1-\varphi^2} + \rho^2(\mathcal{B}^*) + \frac{\varsigma_2^2 \eta^2}{(1-\eta^2)^2} & < \varsigma_2.
		\end{aligned}
	\end{equation}
	
\end{theorem}
\begin{proof}
	From the system model (\ref{coupled model}), we obtain 
	$$
	\begin{gathered}
	(1 - {x_i^V}^* - {x_i^I}^*)\sum_{j \in \mathcal{N}_i^D} \beta_{ij} {x^I_j}^* - \delta_i  {x_i^I}^* = 0, \\
	(\gamma_i(k) + \theta_i(k) ){x_i^V}^* = \delta_i {x_i^I}^* + \theta_i(k) (1-{x_i^I}^*).
	\end{gathered}
	$$
	New system variables are denoted by $\hat{e}_i^I(k) = x^I_i(k) - {x_i^I}^*$ and $\hat{e}_i^V(k) = x^V_i(k) - {x_i^V}^*$. Then, the iteration for $\hat{e}_i^I(k)$ and $\hat{e}_i^V(k)$ can be obtained as
	\begin{align}\label{endemic-y}
		&{\hat{e}_i^I}(k+1) = x^I_i(k+1) - {x_i^I}^* \nonumber\\
		& \ = x^I_i(k) + (1 - x^I_i(k) - x^V_i(k))\sum_{j \in \mathcal{N}_i^D} x^I_j(k) \nonumber\\
		& \ \quad - \delta_i x^I_i(k) - {x_i^I}^* \nonumber\\
		& \ = \hat{e}_i^I(k)  + (1 - \hat{e}_i^I(k) - {x_i^I}^* - {\hat{e}^V_i}(k) - {x_i^V}^*) \nonumber\\
		& \ \quad \times \sum_{j \in \mathcal{N}_i^D} \beta_{ij} (\hat{e}^I_j(k) + {x_j^I}^*) - \delta_i (\hat{e}_i^I(k) + {x_i^I}^*) \nonumber\\
		& \ = (1 - \delta_i - \sum_{j \in \mathcal{N}_i^D} \beta_{ij} {x_j^I}^*) {\hat{e}^I_i}(k)- \sum_{j \in \mathcal{N}_i^D} \beta_{ij} {x_j^I}^* \hat{e}_i^V(k) \nonumber\\
		&\ \quad + (1 - x^I_i(k) - x^V_i(k)) \sum_{j \in \mathcal{N}_i^D} \beta_{ij} \hat{e}^I_j(k),
	\end{align}
	\begin{align} \label{endemic-xe}
		&{\hat{e}_i^V}(k+1) = {\hat{e}^V_i}(k) + \delta_i x^I_i(k) - \gamma_i(k) x^V_i(k) \nonumber\\
		& \quad\quad\quad\quad\quad + \theta_i(k) (1- x^I_i(k) - x^V_i(k)) \nonumber\\
		&\ = {\hat{e}^V_i}(k) + \delta_i(\hat{e}_i^I(k) + {x_i^I}^*) - \gamma_i(k) ({\hat{e}^V_i}(k) + {x_i^V}^*) \nonumber\\
		& \quad + \theta_i(k) (1- \hat{e}_i^I(k) - {x_i^I}^* - \hat{e}_i^V(k) - {x_i^V}^*) \nonumber\\
		&\ = (1 - \gamma_i(k) - \theta_i(k)) {\hat{e}^V_i}(k) + (\delta_i - \theta_i(k)) \hat{e}_i^I(k).
	\end{align}
	Note that (\ref{endemic-y}) and (\ref{endemic-xe}) can be rewritten in the following compact form:
	\begin{equation}
		\left[\begin{array}{c}
			\hat{e}^I(k+1) \\
			\hat{e}^V(k+1)
		\end{array}\right]=\left[\begin{array}{cc}
			 \mathcal{F}_{11}(k) & \mathcal{F}_{12} \\
			 \mathcal{F}_{21}(k) & \mathcal{F}_{22}(k)
		\end{array}\right] \left[\begin{array}{c}
		\hat{e}^I(k) \\
		\hat{e}^V(k)
		\end{array}\right],
	\end{equation}
	where 
	$$
	\begin{aligned}
	\mathcal{F}_{11}(k) &= I_n - \Delta + \operatorname{diag}\left(x^S(k)\right)B - \mathcal{B}^*, \ \mathcal{F}_{12} =  -\mathcal{B}^*, \\ 
	\mathcal{F}_{21}(k) &= \Theta(k) - \Delta,\ \mathcal{F}_{22}(k) = I_n-\Gamma(k)-\Theta(k).
	\end{aligned}
	$$ 
	Denoting $\xi(k) = \left[\hat{e}^I(k)^\top \quad \hat{e}^V(k)^\top\right]^\top$, one obtains
	\begin{equation} \label{combine-endemic}
		\xi(k+1) = F(k) \xi(k),
	\end{equation}
	where $F(k) = \left[\begin{array}{cc}
		\mathcal{F}_{11}(k) & \mathcal{F}_{12} \\
		\mathcal{F}_{21}(k) & \mathcal{F}_{22}(k)
	\end{array}\right]$. 
	
	First, we consider the stability of the system
	\begin{equation} \label{test}
	x(k+1) = \mathcal{F}_{11}(k) x(k).
	\end{equation}
	By (\ref{range}) and the fact that $x^S(k) \in [0, 1]^n$, we can find that $\forall k \geq 0, \|\mathcal{F}_{11}(k)\|_{\infty} < 1$. Moreover, by (\ref{phi}), we have $0<\varphi<1$. Thus similar to the proof of Lemma \ref{lemma-1}, system (\ref{test}) is exponentially stable. Therefore, given any $\varsigma_1 > 0$, there exists a unique $\mathcal{P}_1(k) > 0$ satisfying \cite{rugh1996linear}
	\begin{equation} \label{Lyapunov-1}
		\mathcal{F}_{11}^{\top}(k) \mathcal{P}_1(k+1) \mathcal{F}_{11}(k)-\mathcal{P}_1(k) = -\varsigma_1 I_n,
	\end{equation}
	and $\mathcal{P}_1(k)$ is given by
	\begin{equation}
		\mathcal{P}_1(k) = \varsigma_1\sum_{j=0}^{\infty} (\mathcal{F}_{11}^{\top}(k))^j \left(\mathcal{F}_{11}(k)\right)^j.
	\end{equation}
	We then have 
		\begin{align} \label{p1k}
		\!\left\| \mathcal{P}_1(k) \right\|_{\infty} &\leq \varsigma_1 \sum_{j=0}^{\infty} \left\| (\mathcal{F}_{11}^{\top}(k))^j \left(\mathcal{F}_{11}(k)\right)^j\right\|_{\infty} \nonumber\\
		&\leq \varsigma_1 \sum_{j=0}^{\infty} \left\| \mathcal{F}_{11}^{\top}(k)\right\|_{\infty}^j \left\|\mathcal{F}_{11}(k)\right\|_{\infty}^j \leq \frac{\varsigma_1}{1-\varphi^2}.
		\end{align}
	Based on Assumption \ref{assumption2}, $x(k+1) = \mathcal{F}_{22}(k) x(k)$ is also exponentially stable. Similarly, given any $\varsigma_2 > 0$, there exists a unique $P_2(k) > 0$ satisfying
	\begin{equation} \label{Lyapunov-2}
		\mathcal{F}_{22}^{\top}(k) \mathcal{P}_2(k+1) \mathcal{F}_{22}(k)-\mathcal{P}_2(k) = -\varsigma_2 I_n,
	\end{equation}
	and $\mathcal{P}_2(k)$ is given by
	\begin{equation}
		\mathcal{P}_2(k) = \varsigma_2\sum_{j=0}^{\infty} \left(\mathcal{F}_{22}(k)\right)^{2j}
	\end{equation}
	with 
	\begin{equation} \label{p2k}
		\left\| \mathcal{P}_2(k) \right\|_{\infty} \leq \frac{\varsigma_2}{1-\eta^2}.
	\end{equation}
	Consider the Lyapunov function candidate $V(k)\!=\! \xi(k)^\top P(k) \xi(k)$, where $\mathcal{P}(k)=\left[\begin{array}{cc}\mathcal{P}_1(k) & 0 \\ 0 & \mathcal{P}_2(k)\end{array}\right]$. Then, by (\ref{combine-endemic}), (\ref{Lyapunov-1}), and (\ref{Lyapunov-2}), we have 
		\begin{align} \label{Lyapunov-3}
			V(k+1) - V(k) =& \ \hat{e}^I(k)^\top \mathcal{J}_1(k) \hat{e}^I(k) + \hat{e}^V(k)^\top \mathcal{J}_2(k) \hat{e}^V(k) \nonumber\\
			& + 2\hat{e}^V(k)^\top \mathcal{J}_3(k) \hat{e}^I(k),
		\end{align}
	where 
	$$ 
		\begin{aligned}
			\mathcal{J}_1(k) &= -\varsigma_1 I_n + \mathcal{F}_{21}(k) \mathcal{P}_2(k+1) \mathcal{F}_{21}(k), \\
			\mathcal{J}_2(k) &= -\varsigma_2 I_n + \mathcal{F}_{12} \mathcal{P}_1(k+1) \mathcal{F}_{12}, \\
			\mathcal{J}_3(k) &= \mathcal{F}_{12} \mathcal{P}_1(k+1) \mathcal{F}_{11}(k) + \mathcal{F}_{22}(k) \mathcal{P}_2(k+1) \mathcal{F}_{21}(k).
		\end{aligned}
	$$
	Using (\ref{p1k}), (\ref{p2k}) and the Rayleigh--Ritz theorem in \cite[Theorem 10.13]{laub2004matrix}, one obtains that
		\begin{align} \label{Lyapunov-4}
		\hat{e}^I(k)^\top \mathcal{J}_1(k) \hat{e}^I(k) &\leq \left(-\varsigma_1 + \frac{\varsigma_2 \nu^2}{1-\eta^2} \right) \left\|\hat{e}^I(k)\right\|^2, \\
		\hat{e}^V(k)^\top \mathcal{J}_2(k) \hat{e}^V(k) &\leq \left(-\varsigma_2 + \frac{\varsigma_1 \rho^2(\mathcal{B}^*)}{1-\varphi^2} \right) \left\|\hat{e}^V(k)\right\|^2.
		\end{align}
	Moreover, since $\forall x, y \in \mathbb{R}^n$, $2 x^T y \leq x^T x+y^T y$, the last term of (\ref{Lyapunov-3}) can be bounded as
	\begin{equation} \label{Lyapunov-5}
		\begin{aligned}
			&2\hat{e}^V(k)^\top \mathcal{J}_3(k) \hat{e}^I(k) \\ 
			&\!\!\leq  \hat{e}^I(k)^\top\! \left(\mathcal{P}_1(k\!+\!1) \mathcal{F}_{11}(k) \mathcal{F}_{11}^{\top}(k) \mathcal{P}_1(k\!+\!1) + \mathcal{F}_{21}^2(k)\right)\hat{e}^I(k)  \\
			&+ \hat{e}^V(k)^\top\left(F_{12}^2 + F_{22}(k) \mathcal{P}_2^2(k+1) F_{22}(k)\right)\hat{e}^V(k) \\
			&\!\!\leq \! \Big(\frac{\varsigma_2^2 \eta^2}{(1-\eta^2)^2} \!+\! \rho^2(\mathcal{B}^*)\Big) \! \left\|\hat{e}^V\!(k)\right\|^2  \!+\! \Big( \frac{\varsigma_1^2 \varphi^2}{\left(1-\varphi^2\right)^2} \!+\! \nu^2 \Big)\! \left\|\hat{e}^I\!(k)\right\|^2\!.
		\end{aligned}
	\end{equation}
	Finally, under (\ref{theorem2}), substituting (\ref{Lyapunov-4})--(\ref{Lyapunov-5}) into (\ref{Lyapunov-3}), we have
		\begin{align}
			&V(k+1) - V(k) \nonumber\\
			&\leq  \left(\!-\varsigma_2 + \rho^2(\mathcal{B}^*) + \frac{\varsigma_1 \rho^2(\mathcal{B}^*)}{1-\varphi^2} + \frac{\varsigma_2^2 \eta^2}{(1-\eta^2)^2} \right) \left\|\hat{e}^V(k)\right\|^2 \nonumber\\
			&\quad + \left(-\varsigma_1 + \nu^2 + \frac{\varsigma_2 \nu^2}{1-\eta^2} + \frac{\varsigma_1^2 \varphi^2 }{\left(1-\varphi^2\right)^2} \right) \left\|\hat{e}^I(k)\right\|^2 \nonumber\\
			&\leq  \ 0,
		\end{align}
	where $V(k+1) - V(k)\! = 0$ if and only if $\xi(k)\! = 0$. Therefore, (\ref{combine-endemic}) is asymptotically stable for all disease-nonzero initial conditions. 
	
	Further, denote by $\hat{e}_i^o(k) = o_i(k) - o_i^*$. From the system model (\ref{coupled model}), the iteration for $\hat{e}^o(k)$ can be obtained as 
	\begin{equation}
		\hat{e}^o(k+1) = (I_n-\Phi)\left(\widetilde{W}(o(k)) + \mathcal{O}^* \right) \hat{e}^o(k).
	\end{equation}
	By (\ref{range-1}) and the row-stochasticity of $\widetilde{W}(o(k))$,  we can find that $\| (I_n-\Phi)\left(\widetilde{W}(o(k)) + \mathcal{O}^* \right) \|_{\infty} < 1, \forall k \geq 0$. Thus similar to the proof of Lemma \ref{lemma-1}, $\hat{e}^o(k)$ is globally exponentially stable. Hence, $z^*$ is asymptotically stable for all disease-nonzero initial conditions.
\end{proof}

The above theorem demonstrates that when $R_o^V > 1$, the state of system (\ref{coupled model}) converges to an endemic equilibrium under certain conditions. We should point out that the coupled SIV-opinion dynamics is complicated especially when $R_o^V\! >\! 1$, and our results may not give the full picture. Note that the condition in Theorem \ref{main result-health} is only sufficient to guarantee the convergence of disease-free equilibrium. This thus means that even when $R_o^V > 1$,  the coupled system (\ref{coupled model}) may still have healthy equilibria. Similar issues remain unsolved for other related problems studied in, e.g., \cite{she2022networked,lin2021discrete,pare2020analysis}. Hence, analyzing the existence of a larger healthy/endemic boundary or lack thereof (e.g., locally stable healthy and endemic equilibria coexist in one system with their own attractive region) remains a research direction for future work.

\section{Simulations} \label{simulation}
In this section, we demonstrate how the proposed coupled SIV-opinion model (\ref{coupled model}) can be used to simulate an epidemic spreading process, and illustrate the effectiveness of the derived theoretical results on a real-world large-scale network. 

\subsection{Real-world Network}
We consider an epidemic process spreading over a network of $n = 46$ communities, where each community represents a prefecture of Japan (except for Kumamoto, due to the lack of statistics). Both the physical network for disease spreading and the social network for opinion evolution satisfy Assumption \ref{connection}, but the network structures (i.e., the links) are different because of the distinct spreading patterns in the real world. 

\begin{figure} [t!]
	\centering
	\subfloat[\label{fig1}]{
		\includegraphics[scale=0.50]{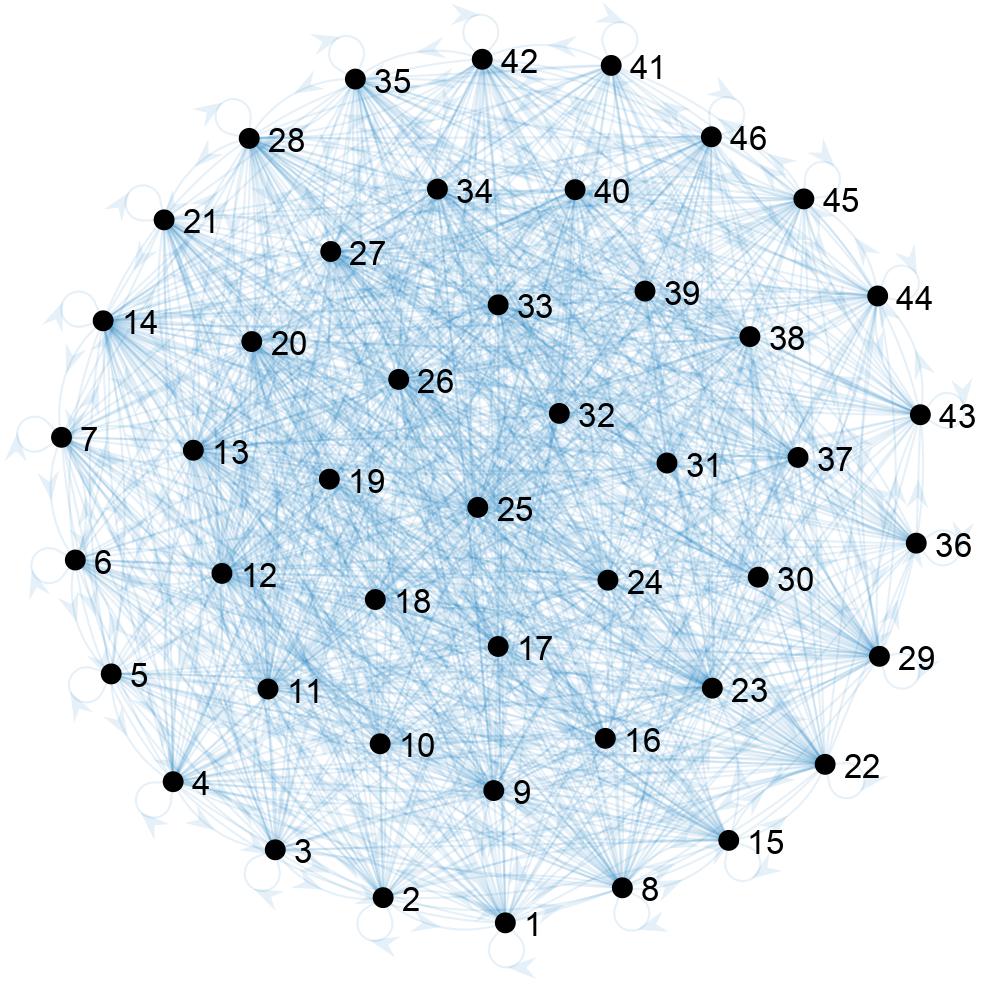}}
	\subfloat[\label{fig2}]{
		\includegraphics[scale=0.58]{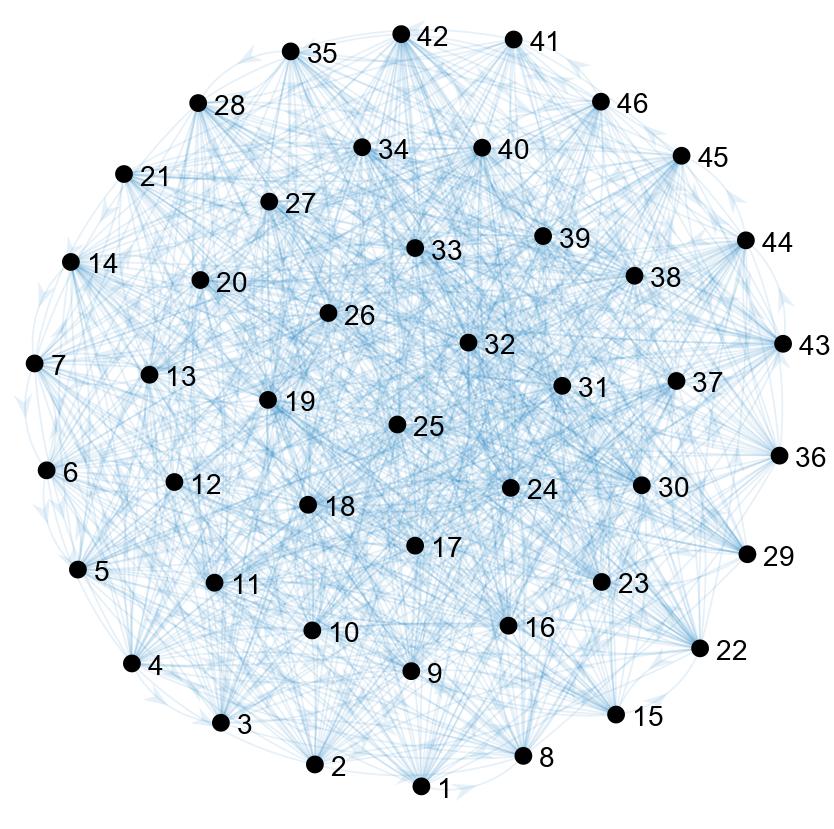} }
	\caption{Network structures. (a) Physical interactions. (b) Opinion interactions.}
\end{figure}


In the physical network depicted in Fig.~\ref{fig1}, the links signify human mobility and migration between prefectures. This network is established using the statistics from the Eighth National Survey on Migration published by the National Institute of Population and Social Security Research in Japan \cite{national2019}.
Note that the original physical network $\bar{\mathcal{G}}_P$ is a complete directed graph, i.e., the original adjacency matrix $\bar{B}$ is asymmetric and strictly positive because population migration exists between any pair of prefectures. However, the large and small entries of $\bar{B}$, representing busy routes and almost deserted routes respectively, differ by several orders of magnitude. In this work, we remove the low weight edges to simplify the network structure (i.e., setting the entries of $\bar{B}$ below a threshold to zero). The irreducibility of the new obtained adjacency matrix $B$ is guaranteed by the fact that the neighbouring areas always communicate sufficiently in Japan. Thus, we obtain a strongly connected subgraph $\mathcal{G}_D$ of $\bar{\mathcal{G}}_P$. Furthermore, the recovery rate matrix $\Delta$ is derived from the Physician Maldistribution Index of Japan in 2022, provided by  the Ministry of Health, Labour and Welfare \cite{ministry2022}, which reflects how health resources are distributed across population and geographical area.


In the social network depicted in Fig.~\ref{fig2}, the links signify individual opinion communication between prefectures. Thanks to the developed Internet and social networking services, individuals from different regions can communicate with each other almost equally. Therefore, we simulate this prefectural social model using the Watts-Strogatz model \cite{watts1998collective} to generate a small-world network with parameters $n$, $d$, and $c$ representing the network size, the average degree, and the clustering coefficient, respectively. We set $n=46$, $d=10$ and $c=0.5$ to generate the network in this section. In such a small world network, most communities are not neighbors of one another, but most communities can be reached from every other communities by a small number of hops, which reflects the sociological phenomenon of six degrees of separation \cite{guare2016six}. 

Finally, recalling Section \ref{section3A}, we consider the well posedness of the simulation data and parameters. Considering that $\theta_i(o_i(k))$ and $\gamma_i(o_i(k))$ are monotonic functions, we choose $\theta_i(o_i(k)) = 0.2 + 0.3o_i(k)$ and $\gamma_i(o_i(k)) = 0.4 - 0.4o_i(k)$. The original physical adjacency matrix $B$ is normalized to satisfy $B \boldsymbol{1}_n = 0.5 \boldsymbol{1}_n$. Moreover, to clarify in the figures to be shown, we display the dynamics of 5 randomly selected communities from the whole network, along with the average values of all communities (shown as a thick black dotted line).

\begin{figure} [t!]
	\centering
	\subfloat[\label{1a}]{
		\includegraphics[scale=0.32]{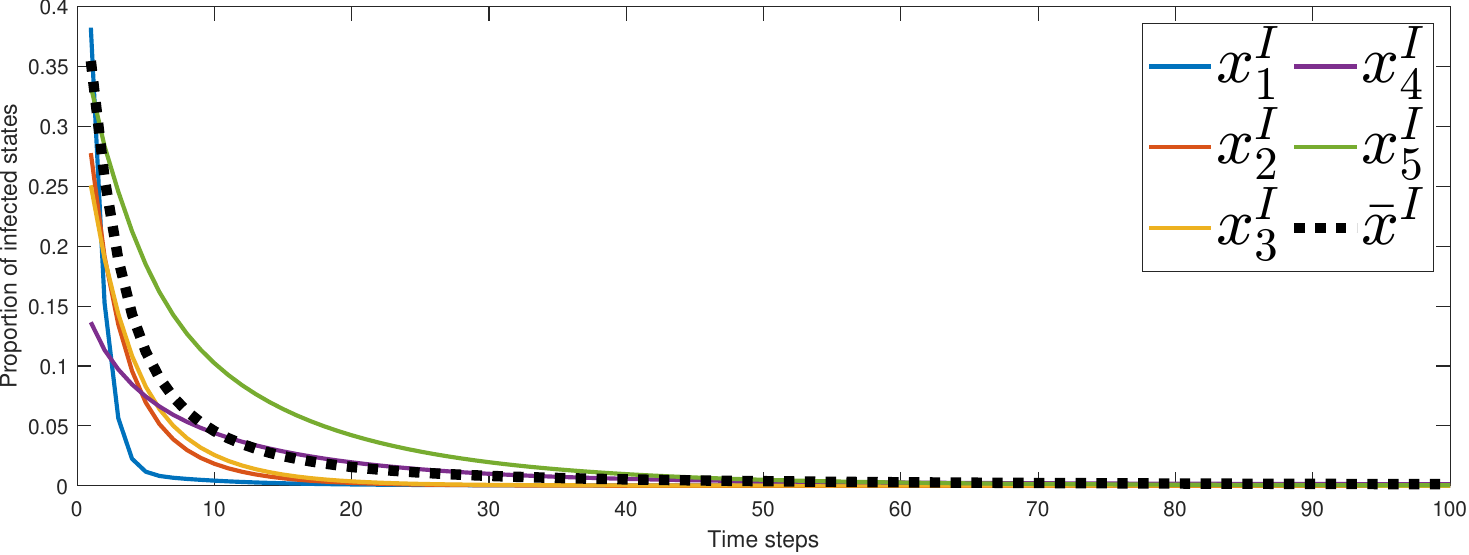}}
	\\
	\subfloat[\label{1b}]{
		\includegraphics[scale=0.32]{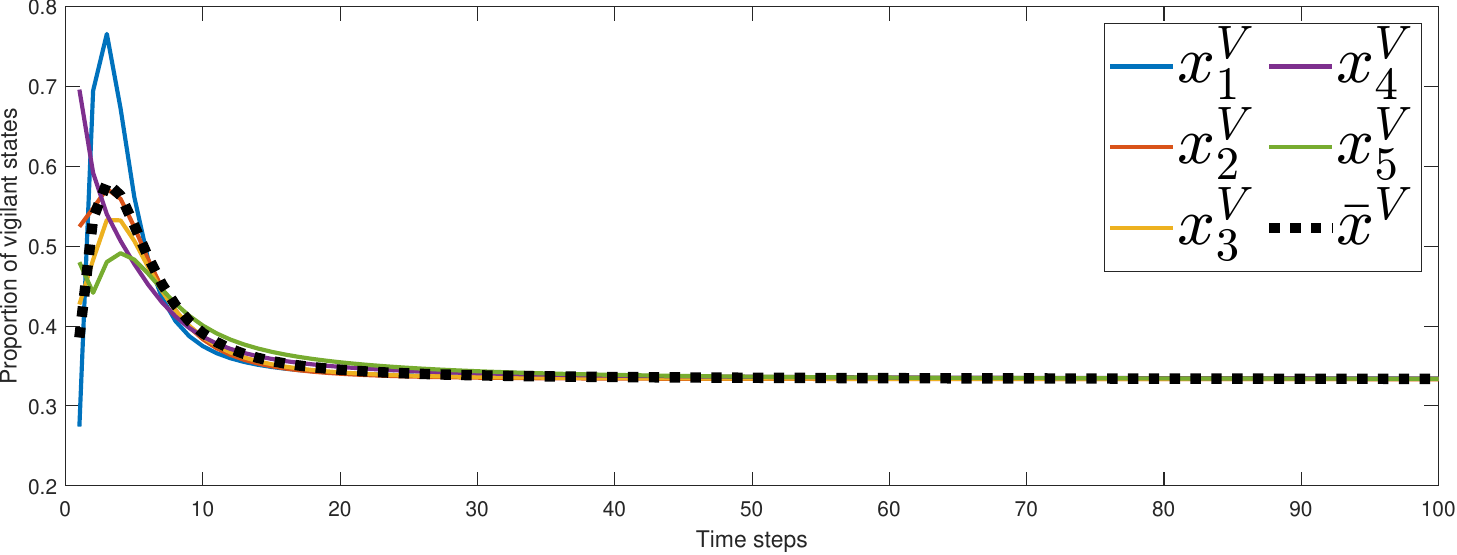} }
	\\
	\subfloat[\label{1c}]{
		\includegraphics[scale=0.32]{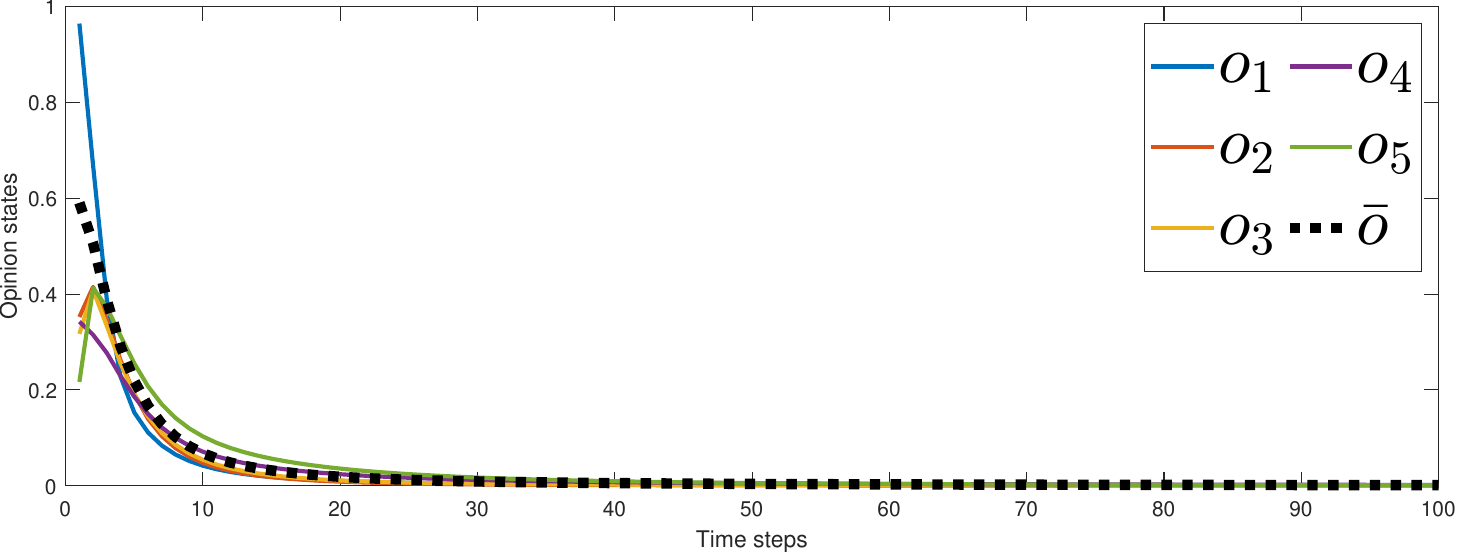} }
	\caption{Under a mild epidemic with $R_o^V = 0.9956$, the evolution of the coupled SIV-opinion system for the $46$ communities network in Fig.~\ref{fig1} and \ref{fig2}. (a) The infected states converge to zero. (b) The vigilant states converge to $0.3333$. (c) The opinion states reach consensus and converge to zero.}
	\label{fig3} 
\end{figure}

\begin{figure}
	\begin{center}
		\includegraphics[width=3.2in]{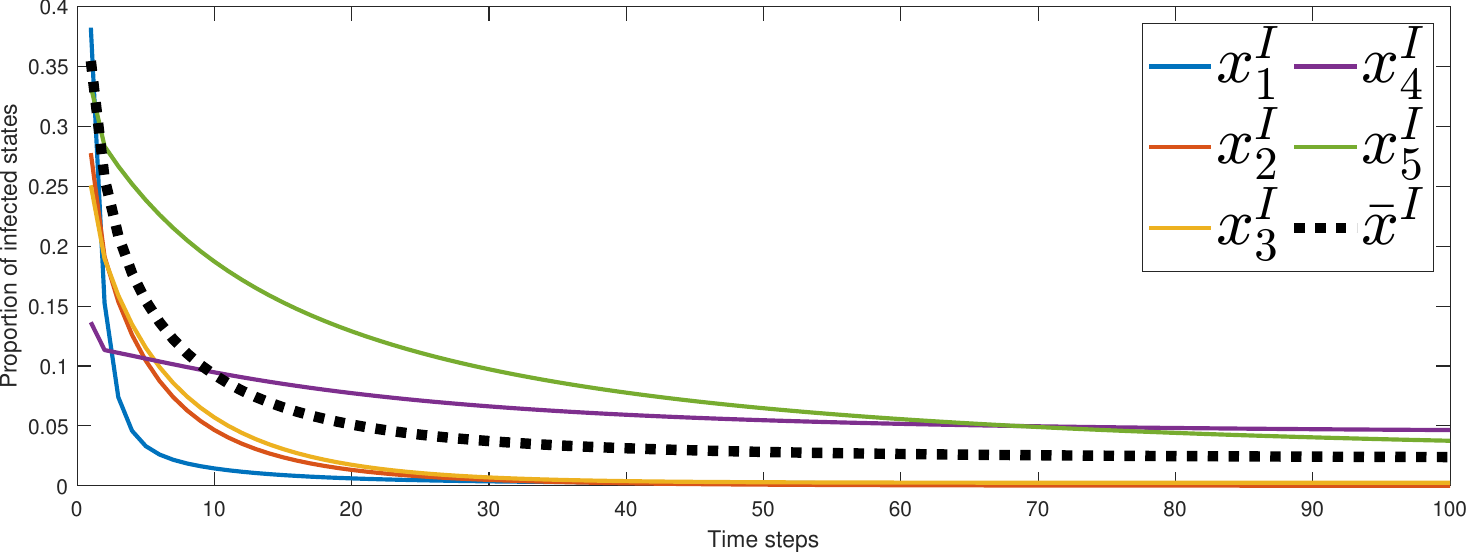}    
		\caption{Under the same mild epidemic as Fig.~\ref{fig3}, the evolution of the infected populations under the SIS model without the effect of opinions.}  
		\label{fig4}                                 
	\end{center}                                 
\end{figure}

\subsection{Mild Epidemics}
First, we simulate the evolution of a mild epidemic with low infectivity, using a scaled adjacency matrix of $0.4B$. Then according to Definition \ref{definition}, we obtain the SIV-opinion reproduction number of $R_o^V = 0.9956$. The initial epidemic-opinion states are generated randomly following Assumption \ref{assumption1}. As Lemma \ref{lemma-1} implies, the opinions of all communities finally converge to a consensus that the epidemic is not serious when the epidemic fades away, as shown in Figs.~\ref{1a} and \ref{1c}. Moreover, when $R_o^V \leq 1$, Theorem \ref{main result-health} states that all the communities converge to a health-consensus equilibrium, which can be computed as $(0, 0.3333\boldsymbol{1}_n, 0)$ for this example. From the plots in Fig.~\ref{fig3}, we confirm this theoretical result.

Additionally, for comparison, we examine epidemic spreading without the vigilance induced by opinion, that is, the discrete-time SIS model. Using the same parameter as in the previous case, the reproduction number is computed as $R_0 = 1.3389$ \cite{pare2020analysis}. In this case, all communities converge to an endemic equilibrium, as illustrated in Fig.~\ref{fig4}. In other words, the SIS model will overestimate the severity of epidemics without taking into account the opinion dynamics, which may result in unnecessary panic in real-world epidemic prevention.

\begin{figure} [t!]
	\centering
	\subfloat[\label{2a}]{
		\includegraphics[scale=0.32]{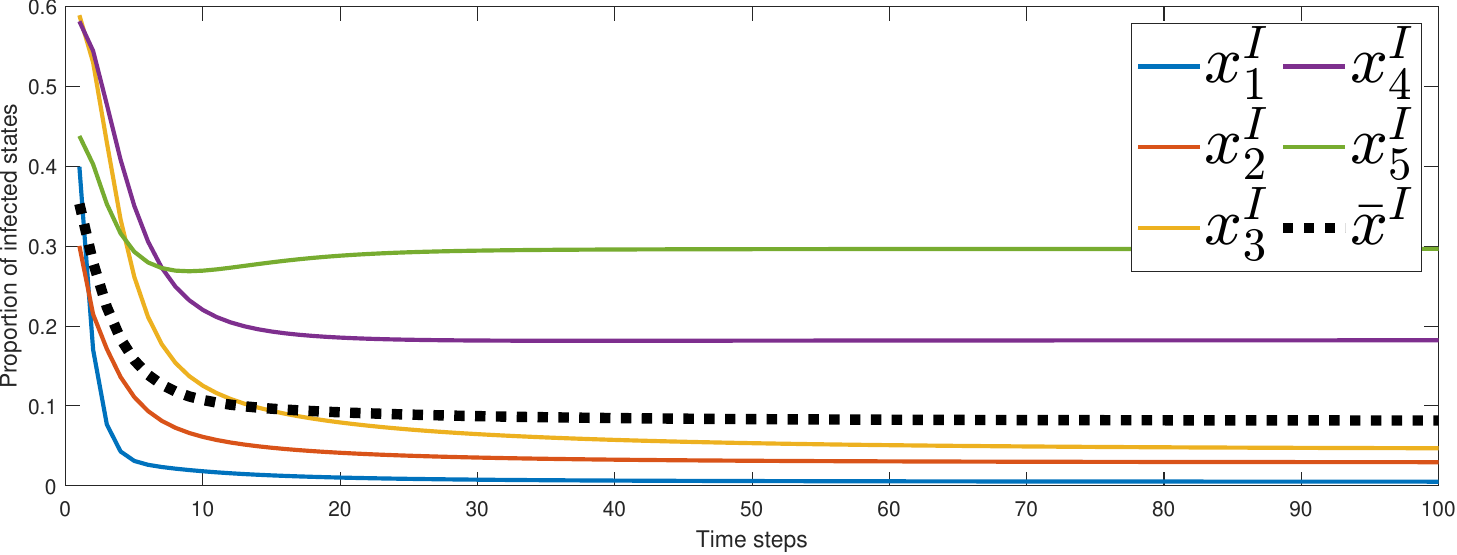}}
	\\
	\subfloat[\label{2b}]{
		\includegraphics[scale=0.32]{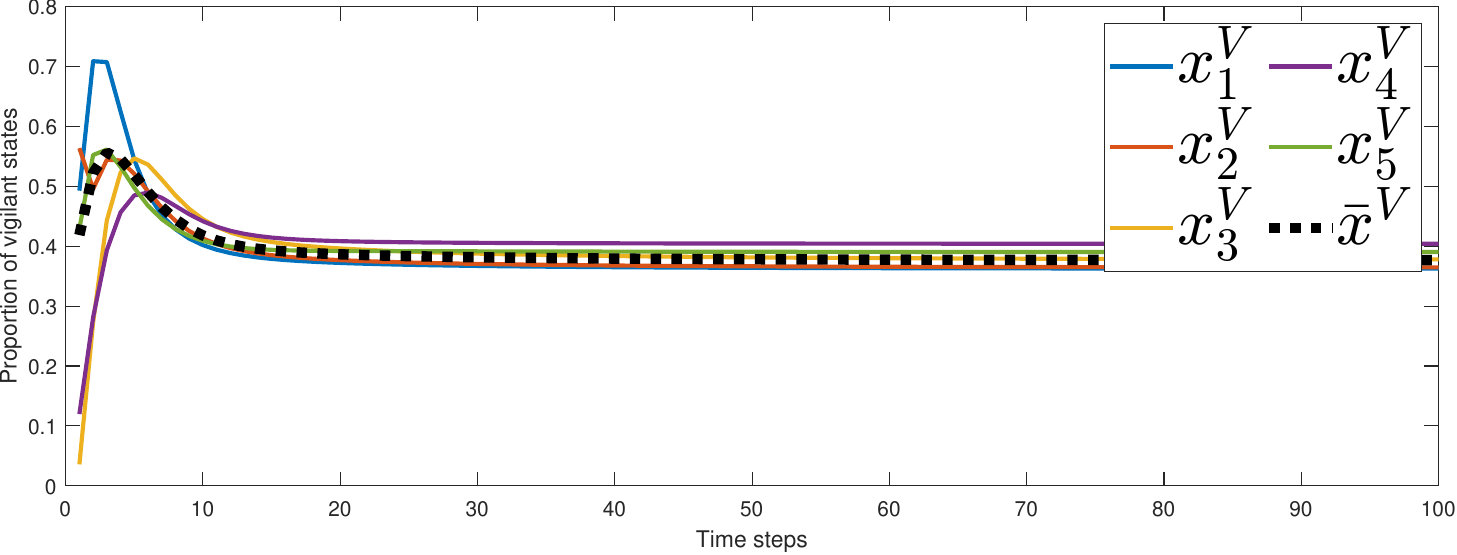} }
	\\
	\subfloat[\label{2c}]{
		\includegraphics[scale=0.32]{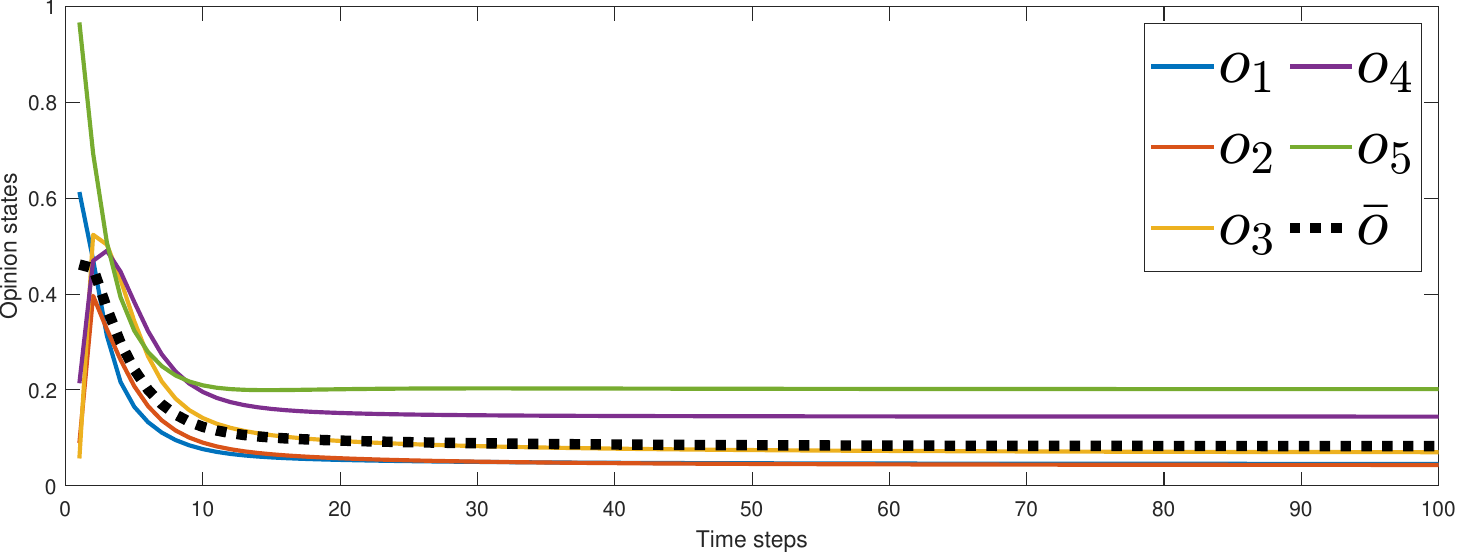} }
	\caption{Under a severe epidemic with $R_o^V = 1.1827$, the evolution of the coupled SIV-opinion system for the $46$ communities network in Fig.~\ref{fig1} and \ref{fig2}. (a) The infected states reach an endemic equilibrium. (b) The vigilant states converge to an equilibrium. (c) The opinion states reach dissensus.}
	\label{fig5} 
\end{figure}

\begin{figure} [t!]
	\centering
	\subfloat[\label{3a}]{
		\includegraphics[scale=0.32]{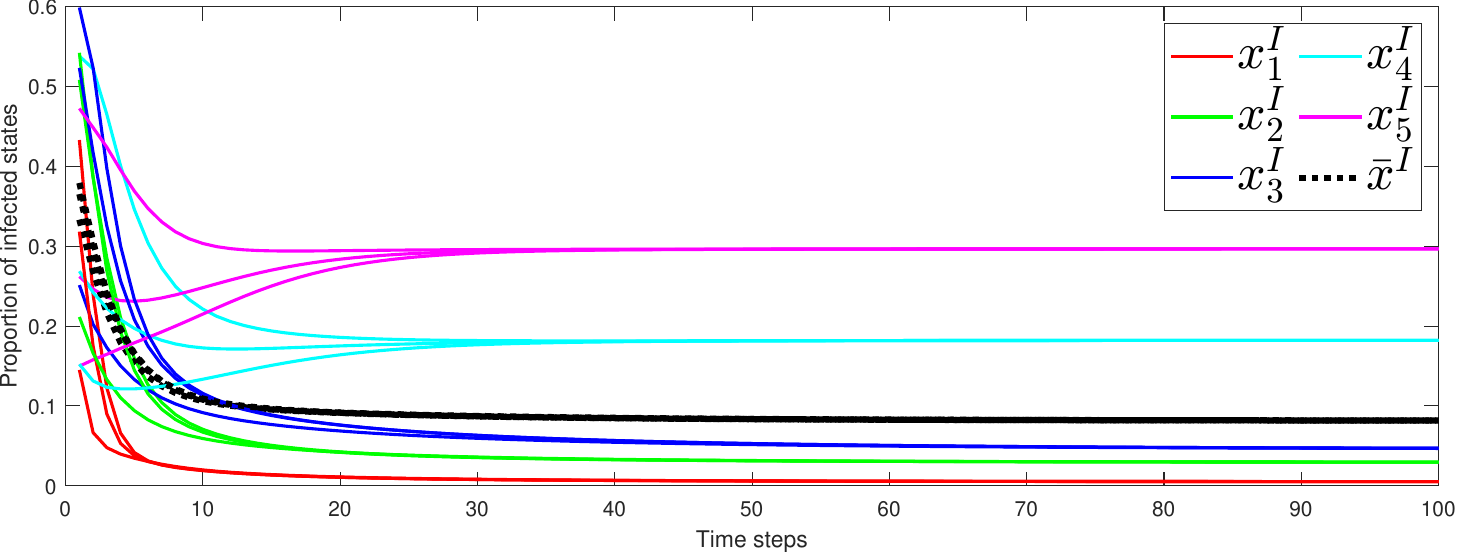}}
	\\
	\subfloat[\label{3b}]{
		\includegraphics[scale=0.32]{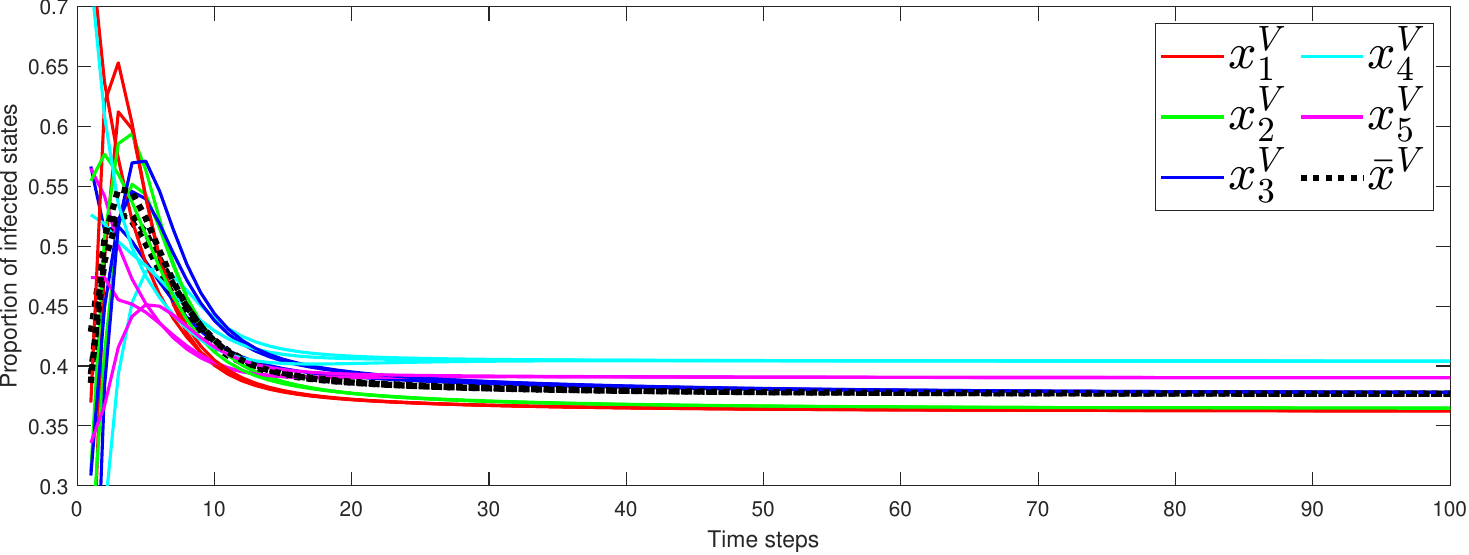} }
	\\
	\subfloat[\label{3c}]{
		\includegraphics[scale=0.32]{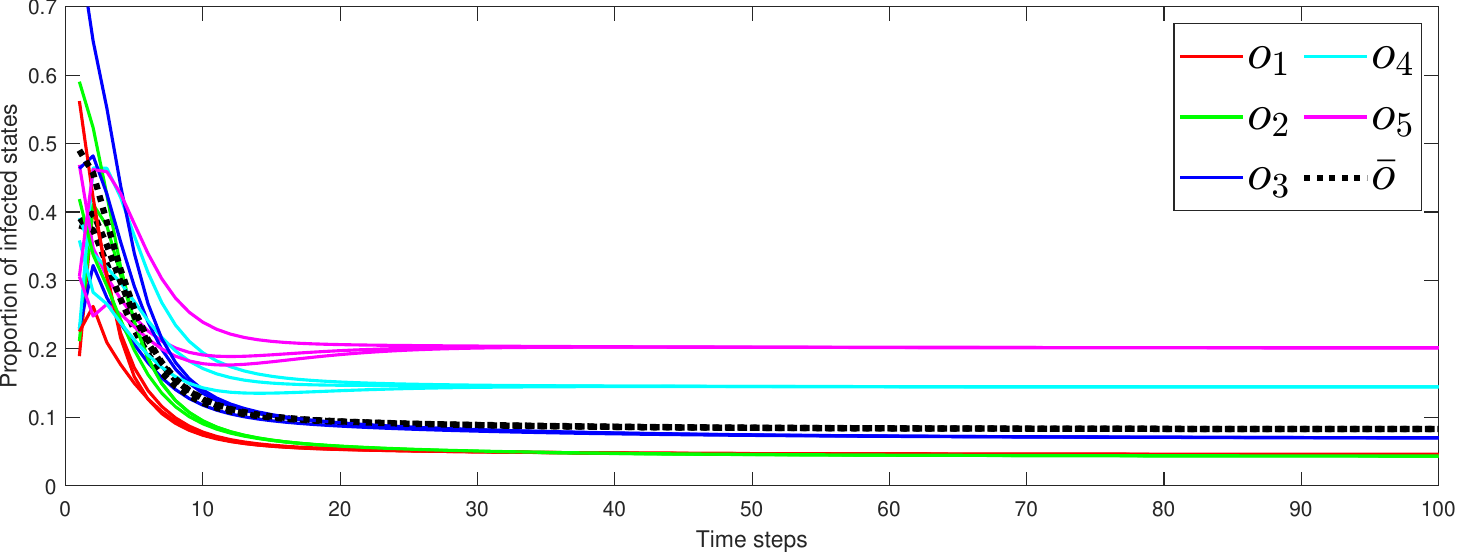} }
	\caption{Under the same condition as Fig.~\ref{fig5}, the evolution of the coupled SIV-opinion system for 3 different initial conditions. The 3 trajectories of each community are depicted by the same color. The system states converge to the same endemic-dissensus equilibrium, independent of the initial condition. (a) Infected states. (b) Vigilant states. (c) Opinion states.}
	\label{fig6} 
\end{figure}

\subsection{Severe Epidemics}
The evolution of a severe epidemic with an adjacency matrix $B$ is illustrated in Fig.~\ref{fig5}. We can obtain the SIV-opinion reproduction number $R_o^V = 1.1827$ following Definition \ref{definition}. As shown in Figs.~\ref{2a} and~\ref{2c}, the coupled SIV-opinion system converges to a dissensus-endemic equilibrium. That is, none of the communities reaches a disease-free state (${x_i^I}^* = 0$) or thinks the epidemic is not a threat (${o_i}^* = 0$), which is consistent with Corollary \ref{corollary-endemic}.

Further, we confirmed through many simulations that the dissensus-endemic equilibrium appears to be unique under different initial conditions. In Fig.~\ref{fig6}, using the same parameters as in Fig.~\ref{fig5}, we start the system with 3 different initial conditions. We can observe that the states converge to the same dissensus-endemic equilibrium, which implies that this equilibrium may have a large region of attraction. We can verify that the equilibrium $z^*$ in Fig.~\ref{fig5} is locally exponentially stable by substituting $z^*$ into Theorem~\ref{theorem-endemic}. Based on $10^5$ Monte Carlo simulations, the stability radius of equilibrium $z^*$ is not less than $0.14$. A challenging question is to find specific conditions for the uniqueness and the region of attraction of dissensus-endemic equilibria theoretically. Theorem \ref{theorem-endemic} provides a sufficient condition, which may be conservative. Tighter conditions remain to be explored in future work.

\subsection{Additional Simulations for Potential Control Strategies}
This part contains additional simulations, which illustrate the feasibility and effectiveness of eradicating epidemics by applying control strategies that affect the reproduction number $R_o^V$, as mentioned in Section \ref{Disease-Free Equilibrium}. Since a severe epidemic may be difficult to eradicate, we consider a moderate epidemic with an adjacency matrix scaled by $0.7B$ with the corresponding reproduction number $R_o^V = 1.0891$.

One of the factors that affect $R_o^V$ is the quality of public health, which can regulate the range of $\psi(o(k))$. For instance, given the same opinions on the epidemic, the people in a community with better sanitary conditions and more established public health policies will be more likely to achieve and maintain the vigilance. This can be characterized by taking larger $\theta$ and smaller $\gamma$. To quantify the levels of public health, let $\epsilon \in [0, 0.3]$ denote the public health parameter. Then, we introduce
\begin{equation} \label{control1}
\begin{aligned} 
	\theta_i(o_i(k)) &= 0.2 + \epsilon + (0.3 -\epsilon)o_i(k), \\
	\gamma_i(o_i(k)) &= 0.4 - \epsilon - (0.4-\epsilon)o_i(k).
\end{aligned}
\end{equation}

Another factor that affects $R_o^V$ is the attention and alertness to the epidemic shown by the governments/administrators of the communities, which can regulate the domain of $\psi(o(k))$. For instance, a cautious government will declare states of emergency and strengthen epidemic prevention publicity in the face of epidemics, which can raise the overall opinions of the public. In other words, residents of an alert community will always maintain a certain level of seriousness about the epidemic. To this end, let $\tau \in [0, 1]$ denote the alertness parameter, and then we extend the opinion dynamics (\ref{opinion-coupled}) under control strategies as
\begin{equation} \label{control2}
	\begin{aligned}
		o_i(k+1)=&\ \tau + (1-\tau)(\phi_i x_i^I(k) + (1-\phi_i)\left[o_i(k), \right.\\
		&+ (1-o_i(k))\sum_{j \in \mathcal{N}_i^O} w_{i j} (o_j(k)-o_i(k))\left.\right]),
	\end{aligned}
\end{equation}
where $\tau = 0$ corresponds to the original one in (\ref{opinion-coupled}).

The equilibria of infected states for $\tau=0$, $\epsilon$ ranging from $0$ to $0.3$, and for $\epsilon=0$, $\tau$ ranging from $0$ to $0.4$ are shown in Fig.~\ref{fig7}. It can be seen that the epidemic states reach an endemic equilibrium without any interventions, while with the increase of $\epsilon$ and $\tau$, the equilibrium converges to zero gradually; it means that all the communities reach the healthy state due to the control strategies.

\begin{figure} [t!]
	\centering
	\subfloat[\label{4a}]{
		\includegraphics[scale=0.32]{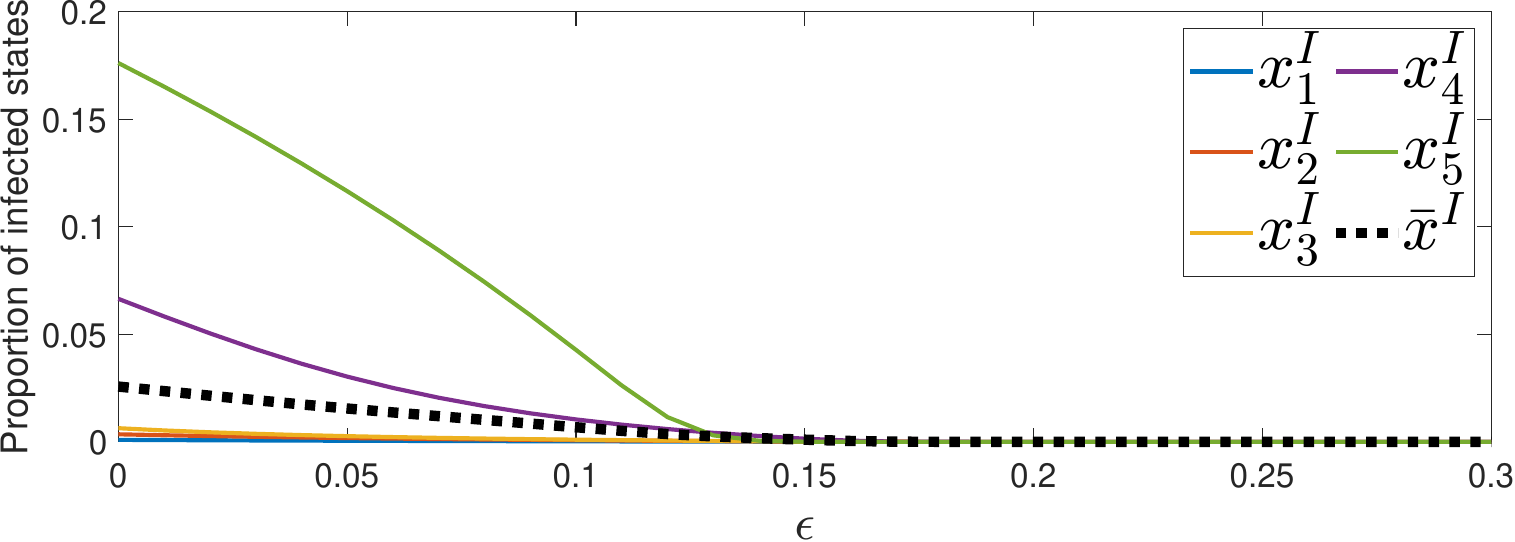}}
	\\
	\subfloat[\label{4b}]{
		\includegraphics[scale=0.32]{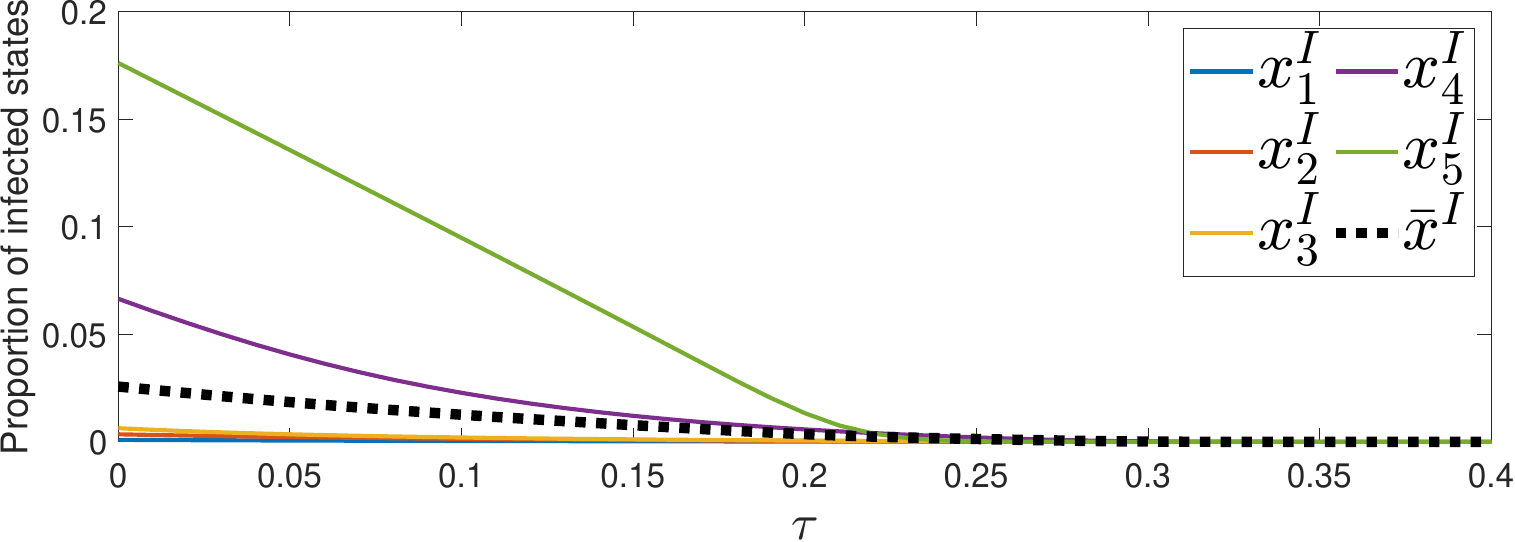} }
	\caption{Under a moderate epidemic with $R_o^V = 1.0891$, the evolution of steady infected states $\lim _{k \rightarrow \infty} x_i^I(k)$ for variable $\epsilon$ and $\tau$. (a) $\tau=0$, $\epsilon$ ranges from $0$ to $0.3$. (b) $\epsilon=0$, $\tau$ ranges from $0$ to $0.4$.}
	\label{fig7} 
\end{figure}

Further, since the above simulation controls every community, a natural question is, whether the control of particular communities can also reach the same target. A direct and qualitative idea is that communities with the largest in-degrees (corresponding to central provinces and transportation hubs in the real world) or the lowest recovery rate (corresponding to under-developed area in the real world) should be critical. We then repeat a similar simulation as in Fig.~\ref{fig7}, but the control strategies (\ref{control1}) and (\ref{control2}) are only applied to the communities with the largest $5$ in-degrees and the lowest $5$ recovery rates, as illustrated in Fig.~\ref{fig8}. The curves show similar convergence to the disease-free equilibria as in Fig.~\ref{fig7}, which verifies the effectiveness of controlling a subset of communities. It is an interesting and significant question to derive the optimal control strategies analytically, including the choices of communities and heterogeneous $\epsilon$ and $\tau$ for each community, which remains a research direction for future work.

\begin{figure} [t!]
	\centering
	\subfloat[\label{5a}]{
		\includegraphics[scale=0.32]{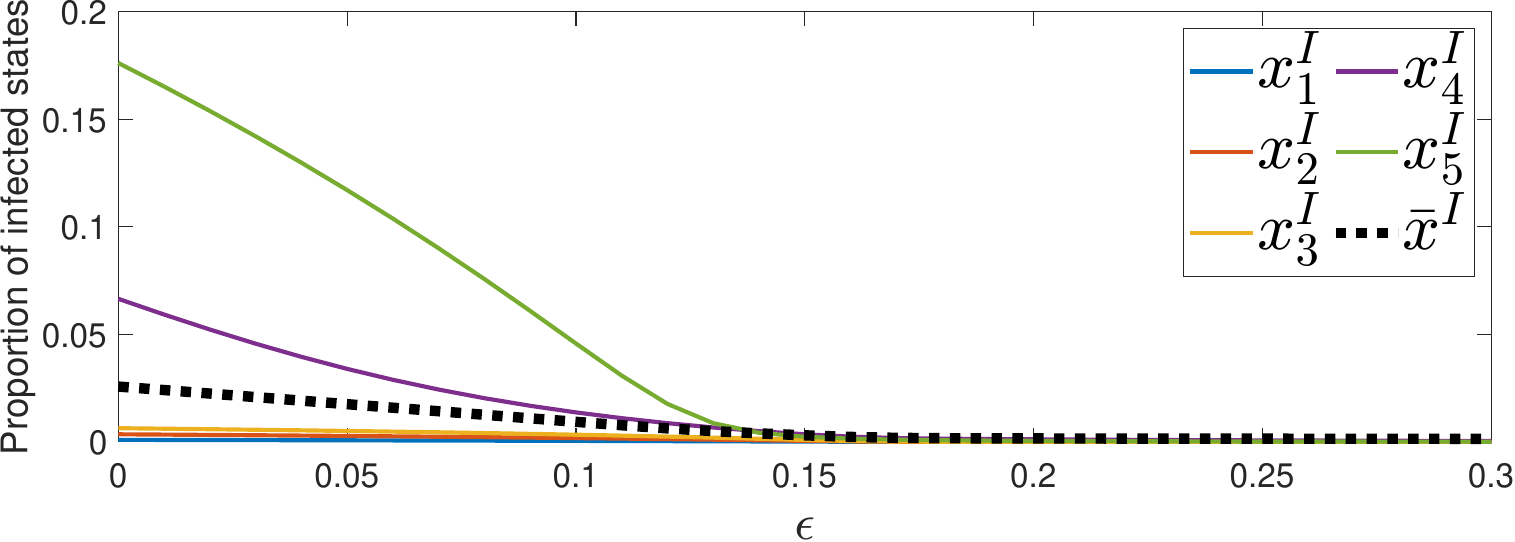}}
	\\
	\subfloat[\label{5b}]{
		\includegraphics[scale=0.32]{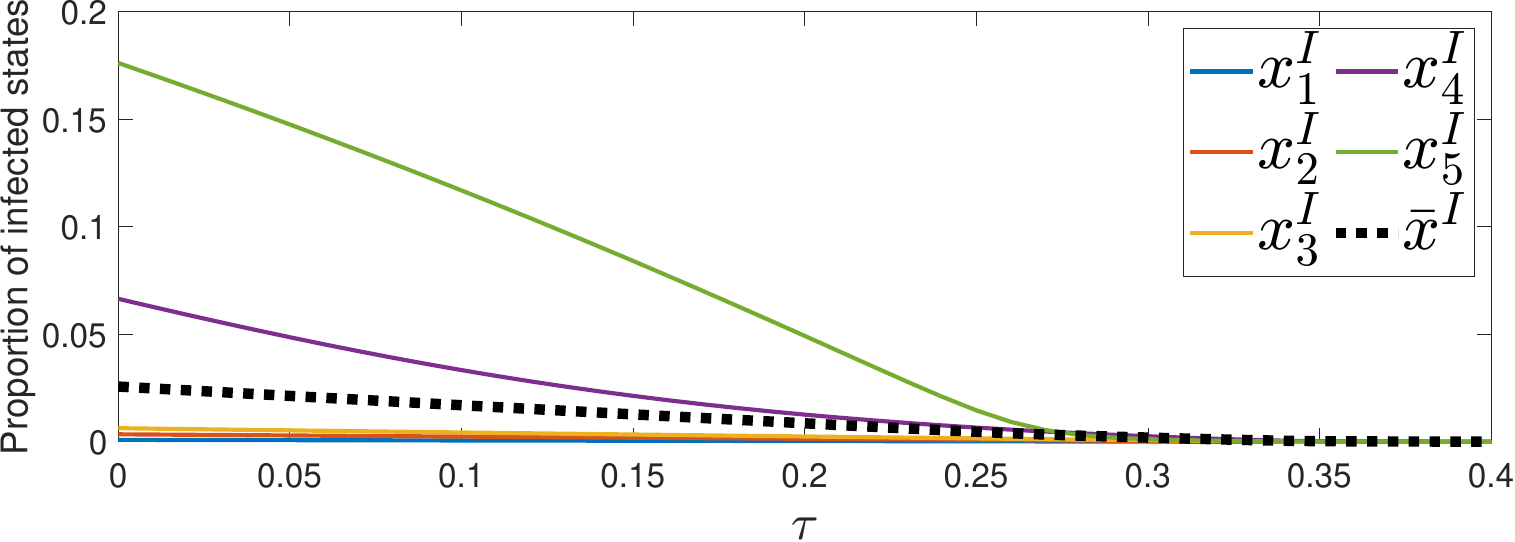} }
	\caption{Under the same condition as Fig.~\ref{fig7}, the evolution of steady infected states $\lim _{k \rightarrow \infty} x_i^I(k)$ with 10 particular communities controlled by variable $\epsilon$ and $\tau$. (a) $\epsilon$ ranges from $0$ to $0.3$. (b) $\tau$ ranges from $0$ to $0.4$.}
	\label{fig8} 
\end{figure}

\section{Conclusion} \label{Section5}
This paper has studied the discrete-time networked SIV epidemic model with polar opinion dynamics. By analyzing this coupled model, we have studied the behavior of epidemic spreading processes, which can be influenced both physically and socially in the real world. In particular, by introducing an SIV-opinion reproduction number, we have obtained sufficient conditions for the stability of disease-free equilibrium and endemic equilibrium. The results reveal the role of opinion dynamics in epidemic spreading, and suggest the possibility of preventing and controlling the epidemic by social interventions. Numerical simulations have been performed to support the theoretical results and some insights on epidemic control. 

For future work, we may consider the rigorous theoretical analysis of optimal control strategies for epidemic eradication, in which many realistic factors, such as resource allocation, partial network control, stubbornness of nodes, heterogeneous or adaptive strategies for each node/edge, are worth exploring. 

\bibliographystyle{IEEEtran}
\bibliography{brief}

\end{document}